\documentclass[12pt]{article}
\usepackage{amsmath}
\usepackage{amsfonts}
\usepackage{mitpress}
\usepackage{comment}
\usepackage[toc,page, title]{appendix}
\usepackage{graphicx}
\usepackage{color}
\usepackage[english]{babel}

\newcommand{\cK}{{\mathcal K}}
\newcommand{\cP}{{\mathcal P}}

\newcommand{\bN}{{\mathbf N}}
\newcommand{\R}{{\mathbf R}}

\newcommand{\AGeo}{{\rm CGeo}}
\newcommand{\AGeol}{\AGeo^\ell}

\newcommand{\CD}{{\rm CD}}
\newcommand{\Dom}{\mathop{\rm Dom}}

\newcommand{\spt}{\mathop{\rm spt}}

\newcommand{\p}{{\partial}}

\newcommand{\LPS}{LPLS}

\newcommand{\NC}{{\rm NC}}
\newcommand{\NE}{{\rm NE}}
\newcommand{\OptTGeoq}{\mbox{\rm OptTGeo}_q^\ell}
\newcommand{\qopt}{{\Gamma^{q}}}
\newcommand{\qopT}{{\Gamma_\ll^{q}}}
\newcommand{\Rc}{{\rm Ric}}
\newcommand{\RCD}{{\rm RCD}}

\newcommand{\TCD}{{\rm TCD}}

\newcommand{\TGeo}{{\rm TGeo}}
\newcommand{\TGeol}{\TGeo^\ell}
\newcommand{\vol}{{\rm vol}}
\newcommand{\vg}{\vol_g}

\newtheorem{theorem}{Theorem}

\newtheorem{corollary}[theorem]{Corollary}

\newtheorem{definition}[theorem]{Definition}
\newtheorem{example}[theorem]{Example}

\newtheorem{lemma}[theorem]{Lemma}

\newtheorem{remark}[theorem]{Remark}

\newenvironment{proof}[1][Proof]{\noindent\textbf{#1.} }{\ \rule{0.5em}{0.5em}}
\newdimen\dummy
\dummy=\oddsidemargin
\begin{document}

\title{A synthetic null energy condition\thanks{%
2020 MSC Classification Primary: 51K10 Secondary: 49Q22 51F99 53Z99 83C99. 
The author's research is supported in part by the Canada Research Chairs program CRC-2020-00289, the Simons Foundation, Natural Sciences and Engineering Research Council of Canada Discovery Grant RGPIN- 2020--04162, and Toronto's Fields Institute for the Mathematical Sciences, where part of this work was performed. He is grateful to Mathias Braun, Nicola Gigli, Christian Ketterer, Clemens S\"amann, Eric Woolgar and three referees for many stimulating exchanges.
\copyright \today\ by the author.
}}
\author{Robert J McCann\thanks{%
Department of Mathematics, University of Toronto, Toronto, Ontario, Canada
\tt mccann@math.toronto.edu} 
}

\maketitle

\begin{abstract}
We give a simplified approach to Kunzinger \& S\"amann's theory of Lorentzian length spaces
in the globally hyperbolic case;
these provide a nonsmooth framework for general relativity.
We close a gap in the regularly localizable setting, by showing consistency of two potentially different notions of timelike geodesic segments used in the literature.  
In the smooth psuedo-Riemannian setting, we show Penrose' null energy condition is equivalent to a variable lower bound on the timelike Ricci curvature.
This allows us to give a nonsmooth reformulation of the null energy condition using the timelike curvature-dimension conditions
of Cavalletti \& Mondino (and Braun).  Although this definition is consistent with the smooth setting,  it proves unstable
relative to the notion of pointed measured convergence for which timelike curvature-dimensions conditions are known to be stable.  
We illustrate this instability 
using a sequence of smooth weighted Lorentzian manifolds-with-boundary that satisfy it,  yet converge to a disconnected pair of timelike related points that violate it in the limit. 
\end{abstract}

From both the mathematical and physical points of view,  a nonsmooth theory of gravity which relaxes or replaces the smooth manifolds required by Einstein's relativity appears highly desirable, as detailed e.g.~in \cite{CavallettiMondino22} and its references.  
The relationship between such a theory and Lorentzian geometry is analogous to the relationship of metric \cite{BuragoBuragoIvanov01} or metric measure \cite{Villani09} geometry to the classical theory of smooth Riemannian manifolds.  Several key steps in this direction have already been taken.  Striking among them are the axiomatization of a theory of Lorentzian length spaces and sectional curvature bounds by Kunzinger \& S\"amann \cite{KunzingerSaemann18}  
(foreshadowed in Kronheimer \& Penrose \cite{KronheimerPenrose67}, Andersson \& Howard \cite{AnderssonHoward98} and Alexander \& Bishop \cite{AlexanderBishop08}), and
 the development of a theory of the Einstein equation and timelike Ricci curvature bounds in this setting by Cavalletti \& Mondino \cite{CavallettiMondino20+} (foreshadowed in work by 
 Mondino \& Suhr \cite{MondinoSuhr23} and myself \cite{McCann20}).  
 By 
 focusing our attention on globally hyperbolic,
regularly localizable spaces we are able to give a simple introduction to this theory, and clarify the equivalence between different notions of timelike geodesics proposed in \cite{KunzingerSaemann18} and \cite{CavallettiMondino20+} \cite{McCann20} respectively.  
In this setting we are able to indicate how to eliminate the dependence of the theory on an auxiliary metrical distance;
compare also to the alternative synthetic frameworks proposed by Minguzzi \& Suhr \cite{MinguzziSuhr22+} and Mueller \cite{Mueller22+}.
Moreover, we resolve a question raised by Cavalletti \& Mondino \cite{CavallettiMondino22}, by giving a nonsmooth reformulation of Penrose' null energy condition (and its weighted variants) --- based on an equivalence we show in the smooth setting --- of the null energy condition to a variable lower bound on the timelike Ricci curvature.   Although technically straightforward, this development is significant both 
because the null energy condition is expected to be satisfied by all forms of matter (unlike 
the strong energy or timelike convergence conditions \cite{Carroll04}) and because in its smooth incarnation, the null energy condition plays a central role in the Penrose singularity theorem \cite{Penrose65a} and sequella surveyed in \cite{Landsman21book} \cite{Landsman21}.   

Although the weak and strong energy conditions are both well-known to imply the null energy condition,
it is worth emphasizing that a variable lower bound of either sign on the timelike Ricci curvature or ---
when the Einstein field equations are satisfied, the stress-energy tensor ---
also implies nonnegativity of null Ricci curvature hence the null energy condition.  Various theorems which rely on the null energy condition,
such as the Penrose singularity theorem \cite{HawkingEllis73} \cite{Penrose65a}, Hawking's monotonicity of area \cite{Hawking72}, 
Galloway's null splitting theorem \cite{Galloway00}, and its consequences such as toplogical censorship \cite{EichmairGallowayPollack13}
therefore also hold under local timelike lower bounds on either the stress energy tensor or the Ricci curvature.

This manuscript is structured as follows.  In the first section
 we give an alternate description of key concepts from Kunzinger \& S\"amann's synthetic geometric framework \cite{KunzingerSaemann18},
 focusing on the globally hyperbolic, regular(ly localizable) case.  Our signed time-separation function $\ell$ plays a central role:
 it captures both causality and chronology, while the non-negative function $\ell_-:=\max\{-\ell,0\}$ of \cite[Remark 2.9]{KunzingerSaemann18} captures only chronology. 
We clarify the relations between different notions of timelike geodesics found in the literature,  focusing on their continuity or potential lack thereof,
and propose a new notion of affinely parameterized
lightlike geodesic.
We also note the auxiliary metric distance function $d$ of \cite{KunzingerSaemann18} 
becomes redundant in the setting of globally hyperbolic Lorentzian length spaces, optionally regular as in  Remark \ref{R:topologies}.
In Section \ref{S:Ricci} we recall how timelike lower Ricci curvature bounds were imposed on such a space by Cavalletti \& 
Mondino \cite{CavallettiMondino20+}, and some subsequent enhancements by Braun~\cite{Braun22.5+} \cite{Braun22.6+}.  In the final section,  we propose a new nonsmooth null
energy-dimension condition,  and show it is consistent with the classical (and Bakry-\'Emery) null energy conditions in the setting of smooth 
(possibly weighted) globally hyperbolic Lorentzian spacetimes.  This answers a question of Cavalletti \& Mondino \cite{CavallettiMondino22}.
Like their timelike curvature dimension condition
upon which it is based,  our null energy condition has structural consequences to be explored in \cite{BraunMcCann23p},
including inequalities and needle decompositions of the volume similar to those used to proved a Hawking type singularity theorem in  \cite{CavallettiMondino22};  
however, unlike their timelike curvature dimension
condition,  
it is unstable in the sense that it need not hold even for a limit of smoothly weighted manifolds that satisfy it.

\section{An alternative approach to Lorentzian length spaces}
\label{S:ghrLLS}

Our theory is set in the Lorentzian (pre)length spaces of Kunzinger \& S\"amann.  
We begin by providing an alternative description of that setting, 
which at a slight cost in generality, simplifies and streamlines the theory by encoding all causality relations 
in a version of time-separation function $\ell$ modifed to take both signs \cite{McCann20}.
Reversing the traditional sign of this function allows us to use the same definitions of, e.g., length and geodesics, for
the time-separation $\ell$ as for a metric space distance $d$.
With this convention the ranges of $\ell$ and $d$ become complementary except on the diagonal.
We use boldface and italics to distinguish those definitions which have {\bf global} importance from those which are used merely {\em locally} or are {\em less central} 
in the present manuscript.

On a set $M$,  a {\bf time-separation function} will therefore refer to a function 
$\ell:M^2 \longrightarrow [-\infty,0] \cup \{\infty\}$ satisfying $\ell(x,x) \le 0$ for all $x,y,z\in M$ plus the triangle inequality
\begin{align}
\ell(x,z) \le \ell(x,y) + \ell(y,z) & {\rm\ if}\ \max\{\ell(x,y),\ell(y,z)\} < \infty.
\label{causal triangle inequality}
\end{align}

These axioms together imply $\ell(x,x) \in \{-\infty,0\}$,  though we shall often assume
$\ell^{-1}(-\infty)$ is empty, in which case $\ell(x,x)=0$, 
and the triangle inequality holds even when $\max\{\ell(x,y),\ell(y,z)\}=\infty$.
They also imply that the negativity and nonpositivity sets $M_\ll^2 = \{ \ell < 0\}$ 
and $M_\le^2 = \{ \ell \le 0\}$  of $\ell$ both denote transitive relations
$\ll$ and $\le$ between members of $M$, with $\le$ being reflexive (i.e. a pre-order) as well.
The points of $M$ are interpreted as spacetime locations or {events}, and the sign of $\ell(x,y)$ 
determines whether or not it is possible for information to pass from $x$ to $y$: we say $x$ lies in the {\bf  causal
past }of $y$ and write $x \le y$ if $\ell(x,y) \le 0$;   we say $x$ lies in the {\bf timelike} or {\bf chronological 
past} of $y$ and write $x \ll y$ if $\ell(x,y) < 0$.   In either case we say $y$ lies in the 
{\bf future} of $x$.  For each point $x \in M$,  one denotes the 
timelike and causal futures by
$$I^+(x):= \ell(x,\cdot)^{-1}([-\infty,0))\ {\rm and}\ J^+(x):= \ell(x,\cdot)^{-1}([-\infty,0])$$  
and the pasts by
$$I^-(y):= \ell(\cdot,y)^{-1}(([-\infty,0))\ {\rm and}\ J^-(y):= \ell(\cdot,y)^{-1}([-\infty,0]).
$$  
In Remark \ref{R:topologies}, 
we describe how the globally hyperbolic Lorentzian length
spaces of 
\cite{KunzingerSaemann18} recalled below are uniquely determined by their time-separation function,
even if the Lipschitz curves within them are not.  

If for all $x,y \in M$, the time-separation function also obeys the {\bf antisymmetry} condition 
\begin{align}
 \label{antisymmetry}
\max\{\ell(x,y),\ell(y,x)\} < \infty &\mbox{\rm\ if and only if}\ x=y,
\end{align} 
then the induced relations $M^2_\ll \subset M^2_\le$ are antisymmetric (so $\le$ is a partial-ordering) and the pair 
$(M,\ell)$ becomes an example of 
a {\em causal space} $(M,\le,\ll)$ 
in Kronheimer \& Penrose' terminology \cite{KunzingerSaemann18} \cite{KronheimerPenrose67}
(but slightly less general because, e.g., if the relations come from a time-separation
function then they enjoy the {\bf push-up property} that $x \le y \ll z$ or $x \ll y \le z$ implies $x \ll z$).
This antisymmetry can always be achieved by replacing events $x \in M$ with equivalence classes $J(x,x)$, as the next lemma shows,
analogously to the standard quotient construction in positive signature \cite{BuragoBuragoIvanov01}.
Antisymmetry holds automatically in the globally hyperbolic causally curve-connected spaces introduced below,
where each equivalence class turns out to consist of a single element.

\begin{lemma}[Antisymmetry via quotienting]
Given a time-separation function $\ell$ on $M$, 
denote $x \sim y$ if and only if $\max\{\ell(x,y),\ell(y,x)\} \le 0$. This defines an equivalence relation on $M$.
Setting 
\begin{equation}\label{tilde ell} 
\tilde \ell(\tilde x,\tilde y) := \inf_{x \in \tilde x, y \in \tilde y} \ell(x,y)
\end{equation}
defines a time-separation function $\tilde \ell$ on the equivalence classes $\tilde x := \{ y \in M \mid y \sim x\}$. Moreover,
$\tilde \ell$ satisfies \eqref{antisymmetry}
on the quotient space $\tilde M := M / \sim$.
\end{lemma}

\begin{proof}
This proof relies heavily on the fact that $\ell(x,y) <\infty$ implies $\ell(x,y) \le 0$.
The relation $x \sim y$ is clearly symmetric.  Transitivity follows from the the triangle inequality \eqref{causal triangle inequality};
reflexivity follows from $\ell(x,x) \le 0$.  To establish the triangle inequality for $\tilde \ell$, it suffices to assume there exist 
$x \in \tilde x, z \in \tilde z$, and $y \sim y' \in \tilde y$ such that $\max\{\ell(x,y),\ell(y',z)\} <\infty$; otherwise there is nothing to prove.
Assuming this existence, since $\ell(y,y') \le 0$ the triangle inequality implies 
\begin{align*}
\ell(x,z) & \le \ell(x,y) +  \ell(y,y') + \ell(y',z)
\\ & \le \ell(x,y) + \ell(y',z);
\end{align*}
 taking infima over $x \in \tilde x, z \in \tilde z$ and $y,y' \in \tilde y$ yields the triangle inequality for $\tilde \ell$.

Since the {\em if} part of the antisymmetry \eqref{antisymmetry} is already established,  we turn to the {\em only if} claim.
Therefore, assume $\max\{\tilde \ell(\tilde x,\tilde y),\tilde \ell(\tilde y,\tilde x)\} <\infty$ for some $\tilde x,\tilde y \in \tilde M$.
Then there exist $x\sim x' \in \tilde x$ and $y\sim y' \in \tilde y$ such that $\ell(x,y) \le 0$ and $\ell(y',x') \le 0$.
Now the triangle inequality implies
\begin{align*}
\ell(y,x) &\le \ell(y,y') + \ell(y',x') + \ell(x',x) \le 0
\end{align*}
hence $x \sim y$ and $\tilde x = \tilde y$ as desired.
\end{proof}

A path $\sigma:A \longrightarrow M$ defined on some interval $A\subset \R$ is said to be {\bf causal} if 
$\ell(\sigma(s), \sigma(t))\le 0$ for all real parameters $s<t$ in $A$; if the inequality is strict the path is said to be {\bf timelike},
whereas if equality always holds 
it is said to be {\bf lightlike} or {\bf null}.
In the non-smooth theory, it is convenient to take causal paths to be future-directed by convention. In contrast
to the {\em curves} we shall presently introduce, no limit on the roughness of {\bf paths} is assumed; when we want to emphasize this potential
lack of continuity,
we call them {\em rough} causal paths.
The {\bf $\ell$-length} (or Lorentzian length) of a causal path $\sigma:[a,b]\subset \R \longrightarrow M$ is defined by
\begin{equation}\label{Lorentzian length}
L_\ell(\sigma) :=  \sup \{ \sum_{i=1}^N \ell(\sigma(s_{i-1}),\sigma(s_{i}))   \mid N \in \bN, a=s_0 < s_1 < \ldots < s_N=b\};
\end{equation}
its magnitude represents the amount of time a particle ages while travelling this path.
The triangle inequality yields $L_\ell(\sigma) \in [\ell(\sigma(a),\sigma(b)),0]$, 
in contrast to the conclusion $L_d(\sigma) \in [d(\sigma(a),\sigma(b)),\infty]$ which holds when a metric distance function $d\ge 0$ is substituted for $\ell$ in \eqref{Lorentzian length}.
For non-compact intervals $A$ we 
 define $L_\ell(\sigma)$ using an increasing sequence of compact subintervals whose union is $A$.
To disambiguate several distinct definitions of $\ell$-geodesic which appear in the literature \cite{KunzingerSaemann18} \cite{McCann20} \cite{CavallettiMondino20+}
we define (timelike) {\bf $\ell$-paths} to be those $\sigma:[0,1]\longrightarrow M$ satisfying 
\begin{align}
\label{Geol} \ell(\sigma(s),\sigma(t)) &= (t-s) \ell(\sigma(0),\sigma(1)) \not\in\{0,\pm\infty\} \qquad \forall\ 0 \le s < t \le 1;
\end{align}
they are affinely parameterized with respect to proper time by convention.
Since $L_\ell(\sigma)=\ell(\sigma(0),\sigma(1))$, the triangle inequality shows as much time elapses following
$\sigma$ as along any other causal path with the same endpoints.
We say $(M,\ell)$ is a {\bf timelike $\ell$-path space} if 
each pair of timelike related points $x \ll y$
are the endpoints for some $\ell$-path $\sigma$. 
We follow \cite{CavallettiMondino20+} by defining $(M,\ell)$ to be {\bf timelike nonbranching} 
unless there exist a pair of distinct $\ell$-paths $\sigma \ne \tilde \sigma$
which coincide on an open interval, $\sigma|_{(s,t)}=\tilde \sigma|_{(s,t)}$ for some $0 \le s < t \le1$,
where $\sigma|_B$ denotes the restriction of $\sigma$ to $B \subset [0,1]$.

Except in the smooth setting \cite{McCann20},
it is not obvious whether $\ell$-paths enjoy any continuity properties.
However, in the (nonsmooth) globally hyperbolic regular Lorentzian geodesic 
spaces described below, 
Corollary~\ref{C:geodesics v extremizers} proves they have a Lipschitz continuous reparameterization.

We call a metric space $(M,d)$ equipped with its metric topology and a time-separation function $\ell$ a {\bf metric spacetime}.
Kunzinger \& S\"amann call a metric spacetime 
in which $\ell_-=\max\{-\ell,0\}$
is lower semicontinuous 
a {\em Lorentzian prelength space} (or {\bf \LPS} hereafter). In an \LPS, $M^2_\ll$ is open.
They say $(M,d,\ell)$ is
{\em causally closed} iff $M^2_\le$ is closed (and {\em locally causally closed} if each $(x,x) \in M^2$ admits a neighbourhood $U \times U$ in which $M_\le^2$ is closed).  Like them, we reserve the term {\bf causal {\bf curve}} for a  {\em nonconstant} causal path which is locally Lipschitz continuous with respect to $d$,
and {\em timelike curve} for a causal curve which is also a timelike path.
They use the term {\em rectifiable} (hereafter {\em $\ell$-rectifiable}) for a causal curve $\sigma$ with nonzero proper time elapsing along each subsegment: i.e.,
$L_\ell(\tilde \sigma)<0$ for the restriction $\tilde \sigma$ of $\sigma$ to each non-degenerate interval $[a,b] \subset A$.
We call a causal curve $\tilde \sigma$ {\bf $\ell$-minimizing} if it minimizes $L_\ell(\sigma)$ among all causal curves which share its endpoints. 
The metric length $L_d(\sigma)$ of any path $\sigma:[a,b] \longrightarrow X$ is defined analogously to  \eqref{Lorentzian length} but with $d$ in place of $\ell$, and 
the path is said to be 
{\bf $d$-rectifiable} 
if $L_d(\tilde \sigma)<\infty$
for the restriction $\tilde \sigma$ of $\sigma$ to each compact interval $[a,b] \subset A$; 
or equivalently, if $\sigma$ is $BV_{loc}$ in $(M,d)$.
Apart from countably many points, any $d$-rectifiable 
path $\sigma$ in $(M,d)$ can be arclength reparametrized monotonically as a $1$-Lipschitz map 
$\tilde \sigma:B \longrightarrow M$  defined on a certain set $B\subset \R$.  Fixing $s \in A$, the countably many omitted points correspond to jump discontinuities of the monotone functions $t \in [s,\infty) \mapsto L_d(\sigma|_{A \cap [s,t]})$ and 
$t \in (-\infty,s] \mapsto L_d(\sigma|_{A \cap [t,s]})$. 
The absence of such jumps is equivalent to the continuity of the path $\sigma$ along with connectedness of the domain $B$ of $\tilde \sigma$.

Similarly to \cite{KunzingerSaemann18} 
we call a metric spacetime 
$(M,d,\ell)$   
{\em non-totally imprisoning} if 
(i) each compact subset 
enjoys a uniform bound 
on the $d$-lengths 
of the causal curves it contains;  we call it {\bf globally hyperbolic} 
if, in addition to being non-totally imprisoning,  (ii) each {\bf causal  diamond} $J(x,y) := J^+(x) \cap J^-(y)$ is compact;
we call it {\bf $\cK$-globally hyperbolic} 
if, in addition to being non-totally imprisoning,  (iii) $J(X,Y) := J^+(X) \cap J^-(Y)$ is compact for each compact $X,Y \subset M$,
where $J^\pm(X) = \cup_{x \in X} J^\pm(x)$.
In the context of the Lorentzian length spaces recalled below,   Burtscher and Garcia-Heveling have shown 
global hyperbolicity becomes equivalent to the existence of a Cauchy time function
or surface~\cite{BurtscherGarcia-Heveling21+}.
We deviate from Kunzinger \& S\"amann's terminology by saying $(M,d,\ell)$ is {\bf causally curve-connected} if $x \le y$ with $x\ne y$ implies the existence of a 
causal curve from $x$ to $y$, and {\bf timelike curve-connected} if $x \ll y$ implies the existence of a timelike curve from $x$ to $y$.

\begin{lemma}[On $d$-rectifiability of causal paths]
\label{L:rectifiability of causal paths}
If a causally curve-connected space $(M,d,\ell)$ is $\cK$-globally hyperbolic,
then each compact subset $X\subset M$ admits a uniform bound on the $d$-length of all (rough) causal paths 
$\sigma:[0,1] \longrightarrow X$ it contains. 
\end{lemma}

\begin{proof}
Since $\cK$-global hyperbolicity yields compactness of $Y= J(X,X)$ from that of $X \subset M$,  non-total imprisonment provides a bound
$C_Y$ on the lengths of (continuous, $d$-rectifiable) causal curves in $Y$.  To derive a contradiction, suppose
$L_d(\sigma)>C_Y$ for some 
rough causal path $\sigma:[0,1]\longrightarrow X$. As in \eqref{Lorentzian length} (but with $d$ replacing $\ell$), there is then a partition 
$0\le s_0 < s_1 <\cdots < s_N \le 1$ for which
$$
\sum_{i=1}^N d(\sigma(s_{i-1}),\sigma(s_i)) > C_Y.
$$
It costs no generality to assume $\sigma(s_{i-1})\ne \sigma(s_i)$ for each $i \in \{1,\ldots, N\}$.
Causal curve-connectedness then yields a ($d$-Lipschitz) causal curve $\tilde \sigma_i$ connecting $\sigma(s_{i-1})$ to $\sigma(s_i)$. This segment lies in $Y$, since its endpoints lie in $X$.  As in \eqref{Lorentzian length},
$L_d(\tilde \sigma_i) \ge d(\sigma(s_{i-1}),d(\sigma(s_i))$. Concatenating these segments therefore yields a ($d$-Lipschitz) causal curve
of length $L(\tilde \sigma)>C_Y$ --- the desired contradiction.
\end{proof}

\begin{remark}[$d$-arclength parameterization of causal paths]\label{R:arclength}
Under the stated hypotheses, Lemma \ref{L:rectifiability of causal paths} asserts any causal path is $d$-rectifiable, hence admits at most countably many discontinuities (and a 1-Lipschitz reparameterization, if continuous).
\end{remark}

A causal curve $\sigma:[0,1]\longrightarrow M$ will be called an {\bf $\ell$-curve} if $L_\ell(\sigma)=\ell(\sigma(0),\sigma(1))$;
by the triangle inequality,  any $\ell$-curve  is $\ell$-minimizing. 
{An  {\bf $M^2_\ll$-geodesic space} will refer to a metric spacetime in which
each pair $(x,y) \in M^2_\ll$ is linked by an $\ell$-curve.} 
Motivated by 3.16--3.18 of \cite{KunzingerSaemann18} we define:

\begin{definition}[Regular Lorentzian geodesic space]
A  {\bf Lorentzian geodesic space} refers to a \LPS\  
in which each pair $(x,y) \in M^2_\le$ with $x \ne y$ is linked by an $\ell$-curve. 
A metric spacetime will be called {\bf regular} if  
no $\ell$-minimizing curve 
$\sigma$ with $\ell(\sigma(0),\sigma(1))<0$ contains nonconstant lightlike subsegments.  
After reparameterization to eliminate any intervals of constancy,  such $\ell$-minimizing curves become timelike.
\end{definition}

We can now address the continuity of $\ell$-paths:


\begin{lemma}[Continuity 
of $\ell$-paths] 
\label{L:continuity of geodesic paths}
Let $(M,d,\ell)$ be a Lorentzian geodesic space in which $\ell_-$ is continuous and all causal diamonds $J(x,y)$ are compact.
Assume no $\ell$-curve $\sigma$ with $\ell(\sigma(0),\sigma(1))<0$ has nonconstant null subsegments.
Then each $\ell$-path $\sigma$ in $M$  is continuous; c.f. \eqref{Geol}. 
\end{lemma}

\begin{proof}  In a regular Lorentzian geodesic space with $\ell_-$ continuous and compact diamonds, 
suppose $\ell(x(s),x(t)) = (t-s) T$ for all $0\le s<t\le 1$, where $T=\ell(x(0),x(1))<0$.

We show left continuity of $x(\cdot)$ at $0< s \le1$. 
Recall $J(x(0),x(s))$ is compact,  while for $t \in (0,s)$ chronology of the path 
implies $x(t) \in I(x(0),x(s)) \subset J(x(0),x(s))$.
As $t$ approaches $s$ from below, $x(t)$ therefore tends to a subsequential limit $p^-$ in the compact set $J(x(0),x(s))$. 
If $p^-=x(s)$ the left continuity of $x(\cdot)$ at $s$ is established, so we assume $p^- \ne x(s)$. 
The continuity of $\ell_-$  implies $\ell(x(0),p^-) =  sT=\ell(x(0),x(s))<0$
(and similarly $\ell(p^-,x(s))\ge 0$ --- which becomes an equality since $p^- \in J(x(0),x(s))$ was noted above).  Concatenating the $\ell$-curves 
joining $x(0)$ to $p^-$ and $p^-$ to $x(s)$ 
yields a causal curve
$\tilde \sigma$ from $x(0)$ to $x(s)$ with Lorentzian length $\ell(x(0),p^-) + \ell(p^-,x(s)) \le sT = \ell(x(0),x(s))$.  
This contradicts the triangle inequality unless $\ell(p^-,x(s))=0$ and $\tilde \sigma$ is an $\ell$-curve with  
timelike separated endpoints
that contains a null segment.  Since the latter would contradict the regularity assumed of $(X,d,\ell)$, we conclude 
the left continuity $p^-=x(s)$ must hold.

A similar argument shows right continuity of $x(\cdot)$ at $s \in [0,1)$ to establish the continuity of $x(\cdot)$.
\end{proof}

Combining Lemmas \ref{L:rectifiability of causal paths} and \ref{L:continuity of geodesic paths} (which provide finite $d$-length and continuity respectively) yields
a corollary whose conclusion (i) seems important for the 
constructions proposed by Cavalletti \& Mondino~\cite{CavallettiMondino20+} \cite{CavallettiMondino22},
and has also been exploited (and credited to us) by Braun \cite{Braun22.5+} \cite{Braun22.6+} in his study of timelike geodesics and curvature dimension conditions.  Its converse (ii) follows from regularity (which implies the $\ell$-minimizers in question are $\ell$-rectifiable, hence timelike) as in Corollary 3.35 of \cite{KunzingerSaemann18}.

\begin{corollary}[Relation of $\ell$-paths to $\ell$-curves and $\ell$-minimizers]
\label{C:geodesics v extremizers}
If $\ell_-$ is continuous on a $\cK$-globally hyperbolic regular Lorentzian geodesic space, 
then:
\\  (i) Every $\ell$-path becomes an $\ell$-curve (hence $\ell$-minimizing)
after a continuous increasing (not necessarily Lipschitz) reparameterization.  
\\ (ii) Conversely, any $\ell$-minimizing curve with timelike
separated endpoints becomes an $\ell$-path after a similar reparameterization.
\end{corollary}


\begin{remark}[Generality]
If globally hyperbolic and regular,  the Lorentzian length spaces of Kunzinger \& S\"amann that we reintroduce in Lemma \ref{L:gh+r LLS} below
satisfy the hypotheses of Lemma \ref{L:continuity of geodesic paths} and Corollary  \ref{C:geodesics v extremizers}.  Those of the lemma %
are more general since, as $\ell$-minimizers are assumed to exist,
the lemma doesn't require a limit curve theorem to hold;  in particular, we don't expect the hypotheses of Lemma \ref{L:continuity of geodesic paths}
to imply $M$ is non-totally imprisoning. 
Moreover, 
Lorentzian length spaces are modelled on open manifolds,
rather than manifolds-with-boundary:  they require $I^\pm(x) \ne \emptyset$ for each $x \in M$.  For this reason, 
a Lorentzian length space cannot be compact, whereas the lemma and its corollary apply equally well to compact and noncompact spaces.
\end{remark}

For the spaces of Corollary \ref{C:geodesics v extremizers},  
we have shown that  (a) $\ell$-paths, (b) $\ell$-curves, and (c) $\ell$-minimizers with timelike separated endpoints differ from each other only by continuous reparameterization.
The advantages of (a) relative to (b) and (c) are that (i) proper-time parameterization plays a crucial role in formulating synthetic timelike sectional \cite{AlexanderBishop08} \cite{KunzingerSaemann18} and Ricci \cite{McCann20} \cite{MondinoSuhr23} \cite{CavallettiMondino20+} curvature bounds; and (ii) it depends only on $\ell$, not on $d$ nor the topology it metrizes.  
As we shall see in Example \ref{E:discontinuous timelike branching}, not all metric spacetimes of interest display the equivalence of (a) (b) and (c); 
noting \cite{Braun22.6+} we therefore define:

\begin{definition}[$\ell$-geodesic]
\label{D:geodesic}
In a metric spacetime $(M,\ell,d)$,  an $\ell$-geodesic will refer to an $\ell$-path which is $d$-continuous. 
\end{definition}

Clearly this definition does not depend on $d$,  except through the topology $d$ induces.
While (b) and (c) also make sense for curves with lightlike separated endpoints, we can extend Definition \ref{D:geodesic} to such paths by first setting
\begin{align}
\TGeol &= \TGeol(M) :=\{\sigma \in C([0,1], M) \mid  \eqref{Geol}\ {\rm holds} \}
\end{align}
and then defining the {\em causal $\ell$-geodesics} $\AGeol(M)$ to be the closure of $\TGeol(M)$ in the space $C([0,1];M)$ of continuous paths metrized by
\begin{equation}\label{uniform metric}
d^\infty(\sigma,\tilde \sigma) := \sup_{s \in [0,1]}d(\sigma(s),\tilde \sigma(s)).
\end{equation}
The continuity of $\ell_-$ required by Lemma \ref{L:continuity of geodesic paths} implies $\sigma \in \AGeol(M)$ satisfies 
$$
\ell(\sigma(s),\sigma(t))=(t-s)\ell(\sigma(0),\sigma(1)) \in \R \qquad \forall 0\le s<t \le 1.
$$
Thus $\sigma \in \AGeol(M)$ is a timelike $\ell$-path unless it is constant or null;  in the latter case we call $\sigma$ a {\em lightlike $\ell$-geodesic},
reserving the unmodified term $\ell$-geodesic for timelike $\ell$-geodesics.

\begin{remark}[Affinely parameterized null paths and role of $d$]\label{R:AGeol}
Like $\TGeol(M)$, the set $\AGeol(M)$  depends on $d$ only through the topology it defines
(since the topology induced by $d^\infty$ on $C([0,1],M)$ depends only on the topology of $M$ induced by $d$).
For the globally hyperbolic Lorentzian length spaces introduced below,  
this topology is uniquely determined by $\ell$, as described in Remark \ref{R:topologies}.
For smooth globally hyperbolic spacetimes (i.e. Lorentzian manifolds satisfying Definition \ref{D:smooth spacetime}), apart from constant curves 
the set $\AGeol(M)$ consists precisely of those
causal geodesics affinely parameterized over $[0,1]$ whose relative interiors contain neither conjugate nor cut points \cite{BeemEhrlichEasley96}.
This justifies the interpretation of lightlike $\ell$-geodesics as the nonsmooth analogs of affinely parameterized smooth null geodesics 
 --- a notion which seems to have remained absent from the nonsmooth literature until now.
\end{remark}


Heuristically,  Kunzinger \& S\"amann's Lorentzian length spaces are metric spacetimes $(M,d,\ell)$ in which each time-separation $\ell(x,y)$ 
is a {\em sharp} lower bound for the Lorentzian length of all causal curves starting from $x$ and ending at $y$; (it is obviously a lower bound by
the triangle inequality).
In this setting, they proved a nonsmooth Avez-Seifert theorem,
asserting that any globally hyperbolic Lorentzian length space is a Lorentzian geodesic space.
In the absence of regularity or 
global hyperbolicity,  the technical definition of a 
{\bf Lorentzian length space} (hereafter {\bf LLS}) becomes cumbersome.  We can instead characterize globally 
hyperbolic Lorentzian length spaces by the following lemma; in its statement the chronological relation $\ll$ plays 
a more conspicuous role than the causal relation $\le$ emphasized in prior formulations.
Note we retain our signed time-separation function $\ell$ 
(which induces both relations $\le$ and $\ll$) instead of the more cumbersome notation $(M,d,\le,\ll,\ell_-)$ of~\cite{KunzingerSaemann18}.

\begin{lemma}[Globally hyperbolic Lorentzian length spaces]
\label{L:gh+r LLS}
Let $(M,d,\ell)$ be a globally hyperbolic metric spacetime.  
Then $(M,d,\ell)$ is a Lorentzian length space if and only if 
(i) it is a timelike curve-connected and
(ii)  $M^2_\ll$-geodesic space, 
in which (iii) $I^\pm(x)$ is non-empty for all $x \in M$, 
(iv) $\ell^{-1}(\infty)$ is open and 
(v) $\ell_- := \max\{-\ell,0\}$ is continuous and real-valued on $M^2$.
\end{lemma}

\begin{proof}
If a globally hyperbolic metric spacetime $(M,d,\ell)$ is a Lorentzian length space,
then it is a locally causally closed, causally and timelike curve-connected (i), 
localizable \LPS\  in which the bound $\ell(x,y)$ on the Lorentzian length of causal curves
from $x$ to $y$ is sharp \cite[Definitions 3.1, 3.16 and 3.22]{KunzingerSaemann18}.
Non-emptiness (iii) of $I^\pm(x)$ is a requirement of localizability.
  Theorems 3.28-3.30 
of the same reference 
then establish (ii) and (v).
Lemmas 1.5-1.6 of \cite{CavallettiMondino20+} now combine to imply causal closedness (iv); c.f. Remark \ref{R:CM20+ L1.5}.

Conversely, if a globally hyperbolic metric space $(M,d,\ell)$ satisfies (i)--(v) it is a (i) timelike curve-connected (v) \LPS, and causally closed by (iv).
The triangle inequality shows the $\ell$-length of any causal curve connecting $x$ to $y$ dominates $\ell(x,y)$.
For $x \ll y$, (ii) asserts this bound is attained.  
The remaining requirements of an LLS are that $(M,d,\ell)$ be {causally curve-connected} and {\em localizable}.
Recall $x \ll y$ are connected by (i) a timelike curve and (ii) an $\ell$-curve.  
We combine (ii)--(v) with global hyperbolicity to argue any distinct $x \le y$ 
are also connected by an $\ell$-curve. 
Since $\ell(x,y)<0$ was handled above, we need only consider $\ell(x,y) =0$.
In this case there exists $z \in I^+(y)$ by (iii) and, since $\ell(y,z)<0$,  an $\ell$-curve $\sigma$ connecting $y$ to $z$ by (ii).
For each $i \in \bN$ let $\sigma_i$ denote the $\ell$-curve connecting $x$ to $\sigma(1/i)$.
Globaly hyperbolicity implies compactness of $J(x,z)$ and also yields a uniform bound $L_d(\sigma_i) < C$.
Letting $\tilde \sigma_i:[0,1]\longrightarrow X$ denote the reparameterization of $\sigma_i$ proportionally to its $d$-arclength,
the generalized Ascoli-Arzela theorem yields a uniform subsequential limiting curve $\tilde \sigma_\infty$
connecting $x$ to $y$ which has $d$-Lipschitz constant $C$ and, by (iv)--(v), is causal.
Since $L_\ell(\sigma) \ge \ell(x,y)=0$ for any causal path joining $x$ to $y$, 
while $0 \ge L_\ell(\tilde \sigma_\infty)$ from causality, it follows that $\tilde \sigma_\infty$ is an $\ell$-curve.

The localizing neighbourhoods and continuous time-separation 
functions required for {\em localizability} are constructed as follows:
given $x \in M$, (iii) yields $x^\pm \in I^\pm(x)$.  The diamond $U_x := I(x^-,x^+)$ is an open neighbourhood of $x$ by (v).
It is contained in the compact set $J(x^-,x^+)$,  hence admits a uniform bound on the $d$-length of all causal curves in $U_x$
by global hyperbolicity.  The restrictions $d',\ell'$ of $d,\ell$ to $U_x \times U_x$  make $(U_x,d',\ell')$ an \LPS\  by (iv)--(v) which also imply the
continuity of $\ell'_-$.  Any $y \in U_x$ admits timelike curves connecting it to $x^\pm$.
This implies $I^\pm(y)\cap U_x$ is nonempty.  The preceding paragraph also shows any $y < z$ in $U_x$ are connected by an $\ell$-curve $\tilde \sigma_\infty$
in $M$.  Since $x^- \ll y \le \tilde \sigma_\infty(s) \le z \ll x^+$ for all $s \in [0,1]$, this $\ell$-curve lies in $U_x$ as required.
\end{proof}

\begin{remark}[Causal closedness, time-separation semicontinuity]\label{R:ell semicontinuity}
Note causal closedness (iv) becomes equivalent to the lower semicontinuity of the time-separation function $\ell$, 
given continuity (v) of $\ell_-$.
\end{remark}

\begin{remark}[$\cK$-global hyperbolicity]
\label{R:CM20+ L1.5}
Any globally hyperbolic Lorentzian length space is $\cK$-global hyperbolic by \cite[Lemma 1.5]{CavallettiMondino20+} \cite{Minguzzi23+},
and Lorentzian geodesic by \cite{KunzingerSaemann18} (or the proof of Lemma \ref{L:gh+r LLS}).  Thus every 
continuous causal path admits a Lipschitz reparameterization by Remark \ref{R:arclength}. 
\end{remark}

\begin{remark}[Regular Lorentzian geodesic spaces]
\label{C:gh+r LLS}
In a regular metric spacetime, Lemma \ref{L:gh+r LLS}(i) follows from (ii). Thus
a globally hyperbolic regular Lorentzian geodesic space in which $\ell$ is lower semicontinuous and $\ell_-$ is real continuous,
is an LLS if and only if (iii) of the same lemma holds (because Lorentzian geodesy implies (ii), 
while the semicontinuity properties of $\ell$ imply (iv)-(v) by Remark \ref{R:ell semicontinuity}).
Thus the main difference is that a smooth manifold with nonempty spatial boundary might be a globally hyperbolic regular Lorentzian geodesic space, but it cannot be an LLS.
\end{remark}

\begin{remark}[Equivalence of topologies and independence of $d$]\label{R:topologies}
Kunzinger \& S\"amann discuss three topologies on a metric spacetime: 
the {\em Alexandrov} (or {\em order}) topology, which is coarsest topology $\mathcal A$ containing the diamonds $\{I(x,y):= I^+(x) \cap I^-(y)\}_{x,y \in M}$, 
the {\em chronological} topology, which is the coarsest topology $\mathcal I$ containing the cones $\{I^\pm(x)\}_{x \in M}$,
and the metric topology $\mathcal D$. In an \LPS\  these are ordered
$\mathcal A \subset \mathcal I \subset \mathcal D$.  They call an \LPS\  {\em strongly causal} if all three topologies coincide.
Theorem 3.26 
of \cite{KunzingerSaemann18} shows that that a globally hyperbolic LLS is strongly causal, and  that a strongly causal LLS is 
non-totally imprisoning.  Since strong causality and compactness of causal diamonds are purely topological properties,  it follows that if $(M,d,\ell)$ and $(M,d',\ell)$ are Lorentzian length spaces in which the metric topologies
$\mathcal D$ and $\mathcal D'$ coincide,  then global hyperbolicity of one implies global hyperbolicity of the other.
Similarly, in a $\cK$-globally hyperbolic Lorentzian geodesic space (or by Remark \ref{R:CM20+ L1.5} in a
 globally hyperbolic Lorentzian length space) Lemma \ref{L:rectifiability of causal paths} shows
regularity becomes equivalent to the assertion that every continuous causal path $\sigma$ satisfying 
$L_\ell(\sigma)=\ell(\sigma(0),\sigma(1))<0$ is free from nonconstant null segments.
Thus we arrive at the notion of a globally hyperbolic (regular) Lorentzian length space which 
does not depend on $d$ at all,  requiring only metrizability of the chronological topology;  
compare \cite{Minguzzi23+} \cite{MinguzziSuhr22+} \cite{Mueller22+}. 
\end{remark}

\section{Synthetic timelike Ricci curvature bounds}
\label{S:Ricci}

We shall often be interested in metric spacetimes $(M,d,\ell)$ which are (a) regular (b) $\mathcal K$-globally hyperbolic (c) Lorentzian geodesic spaces having (d) $\ell$ lower semicontinuous and (e) $\ell_-=\max\{0,-\ell\}$ continuous real-valued.  Such spaces differ from globally hyperbolic regular 
LLS's only in that
$I^\pm(x)$ can be empty for some $x \in M$, thus encompassing some manifolds with spatial boundary and facilitating heredity of (a)-(e) by compact subsets.
We shall also assume 
the topology is (f) separable and (g) complete.

\begin{definition}[Causal geodesic space]
\label{D:causal geodesic space}
Metric spacetimes satisfying (a)-(g) above will be called {\bf causal geodesic spaces} for brevity,
or {\bf proper} causal geodesic spaces if (h) boundedly compact.
\end{definition}

They seem a natural setting for Cavalletti \& Mondino's 
formulation of synthetic timelike lower Ricci curvature bounds, which we now recall along with
refinements due to Braun.  A slightly more restrictive alternative would be to work on closed Lorentzian geodesic subsets of 
(a') regular (b') globally hyperbolic (c') Polish Lorentzian length spaces,  in which case Remark~\ref{R:topologies} shows we can discard $d$ entirely and assume only complete separable metrizability of the topology generated by $\{I^\pm(x)\}_{x \in M}$. 

\subsection{Lifting the geometry from events to fuzzy events}


Given a metric space $(M,d)$, let $\cP(M)$ denote the set of Borel probability measures on such a space, and $\cP_c(M) \subset \cP(M)$ the subset of measures with compact support.  We can think of $\mu \in \cP(M)$ as representing a fuzzy point (or in the metric spacetime setting, a fuzzy event).
 Given a Borel map $G:M^- \longrightarrow M^+$ between two metric spaces
 $(M^\pm,d^\pm)$ and $\mu^- \in \cP(M^-)$,  we denote by $\mu^+ = G_\#\mu^- \in \cP(M^+)$ the measure defined by
 $\mu^+(B) = \mu^-(G^{-1}(B))$ for all $B \subset M^+$.  Letting $\pi^\mp(x^-,x^+) = x^\mp$ 
 denote the projection from $M^- \times M^+$
 onto its left and right factors,  we define 
 $$\Gamma(\mu^-,\mu^+) = \{ \gamma \in \cP(M^- \times M^+) \mid \pi^\pm_\# \gamma = \mu^\pm\}.$$
 Given $p \in [1,\infty)$,  the {\em $p$-Kantorovich-Rubinstein-Wasserstein distance} $d_p$ between $\mu^\pm \in \cP(M)$ defined by
 \begin{equation}\label{d_p}
 d_p(\mu^-,\mu^+) := \inf_{\gamma \in \Gamma(\mu^+,\mu^-)} \left(\int_{M^2} d(x,y)^p d\gamma(x,y) \right)^{1/p} 
 \end{equation}
 is well-known to be a metric on $\cP_c(M)$ provided $(M,d)$ is Polish (i.e. complete and separable), as we henceforth assume.
In this case the infimum defining $d_p$ is attained.  The completion of $\cP_c(M)$ with respect to $d_p$ consists of the measures
$\cP_p(M) \subset \cP(M)$ having moments up to order $p$,  and $d_p$ is well-known to metrize {\bf narrow} convergence 
(against continuous bounded test functions) plus convergence of these moments  \cite{Villani03}.  Moreover, $(\cP_p(M),d_p)$ is Polish \cite{AmbrosioGigli13o}.

In a causal geodesic space,  we set 

\begin{align*}
\Gamma_\le(\mu,\nu) &:= \{\gamma \in \Gamma(\mu,\nu) \mid \gamma[M^2_\le]=1 \}
\\ \Gamma_\ll(\mu,\nu) &:= \{\gamma \in \Gamma(\mu,\nu) \mid \gamma[M^2_\ll]=1 \}
\end{align*}
for $\mu,\nu \in \cP(M)$.  A measure $\gamma \in \cP(M^2)$ is called {\bf causal} if $\gamma[M^2_\le]=1$ and {\bf timelike} if $\gamma[M^2_\ll]=1$.
 Given $q \in (0,1]$,  we define
 \begin{equation}\label{l_q}
\ell_q(\mu,\nu) := \inf_{\gamma \in \Gamma_\le(\mu,\nu)} -\left(\int_{X^2} |\ell(x,y)|^q d\gamma(x,y) \right)^{1/q} 
 \end{equation}
 as in \cite{EcksteinMiller17} \cite{McCann20} \cite{CavallettiMondino20+};
 apart from the causality restriction $\gamma[M^2_\le]=1$, this is analogous to \eqref{d_p}.
  Although $\mu \otimes \nu \in \Gamma(\mu,\nu)$, it is possible that $\Gamma_\le (\mu,\nu)=\emptyset$;
 when this is the case we set $\ell_q(\mu,\nu)=+\infty$.  
 In \cite{EcksteinMiller17} \cite{McCann20} \cite{CavallettiMondino20+} it is shown that $\ell_q$ satisfies
 the triangle inequality \eqref{causal triangle inequality}.  
 When $\mu,\nu \in \cP_c(M)$ have compact support and $\Gamma_\le(\mu,\nu) \ne \emptyset$, from the continuity of $\ell_-$ we find
 $\ell_q(\mu,\nu) \in (-\infty,0]$ and the supremum \eqref{l_q} is attained; for attainment we use the fact that $M^2_\le$ is closed (by Remark~\ref{R:ell semicontinuity} and (c)-(d) of Definition \ref{D:causal geodesic space})
 to ensure $\Gamma_\le(\mu^-,\mu^+)$ is compact in the narrow topology (of convergence against continuous bounded test functions).  It is then easy to see that $\ell_q$ is a time-separation function on $\cP_c(M)$,
 and the induced causality relations on $\cP_c(M)$ are independent of $q \in (0,1]$.  Moreover,  $(\cP_c(M),\ell_q)$ is a timelike $\ell_q$-path space.
By contrast, the next remark shows $(\cP_c(M),d_p,\ell_q)$ cannot be an LPLS  for any $0<q\le 1\le p <\infty$.
  
 \begin{remark}[The $q$-time separation is not semicontinuous]
{Neither $\ell_q$ nor $\min\{\ell_q,0\}$ can be $d_p$-upper semicontinuous:
 when $\ell_q(\mu_\infty,\nu)<0$,  it is easy to construct narrow limits $\mu_j \to \mu_\infty$ 
 for which $\ell_q(\mu_j,\nu)= \infty$ by making small perturbations of $\mu_\infty$ that violate causality by being located in the future rather than the past of $\nu$.}
 \end{remark}
 
  It is also easy to construct examples of $\ell_q$-paths which are neither $d_p$-continuous nor timelike nonbranching;
 the following example helps motivate Definitions \ref{D:geodesic} and \ref{D:dualizability}.
  
  
 \begin{example}[Discontinuous, timelike branching $\ell_q$-paths]
 \label{E:discontinuous timelike branching}
 Fix the standard coordinates on two-dimensional Minkowski space $M=\R^2$.  Let $\ell(x,y) = -|g(y-x,y-x)|^{1/2}$ denote 
 the standard time-separation function (extended as $+\infty$ unless $y$ is in the future of $x$) and $d(x,y) = |x-y|$ the Euclidean distance in the chosen coordinates.
 Then $x(t) = (t,5)$ and $y(t) =(t,t)$ are timelike and lightlike geodesics respectively,   with $\ell(x(s),y(t)) = +\infty$ for all $s,t \in [0,1]$.
 Given $0<r<1$ set 
 $$
 z(t) :=
  \begin{cases} y(0) & {\rm if}\ t\in [0,r), 
 \\ y(1) & {\rm if}\ t \in [r,1].
 \end{cases}
 $$
  The measure $\mu(s) = \frac12 [\delta_{x(s)} + \delta_{z(s)}]$ is an $\ell_q$-path which fails to be weakly continuous.
  Varying $0<r<1$ also shows $(\cP_c(M),\ell_q)$ fails to be timelike nonbranching.
  \end{example}
 
  Let $\qopt = \qopt(\mu^-,\mu^+)$ denote the set of measures that optimize \eqref{l_q},
 and $\qopT(\mu^-,\mu^+) =( \qopt \cap \Gamma_\ll)(\mu^-,\mu^+)$.   To exclude cases where 
 this intersection is empty, 
 we recall the following definition \cite{CavallettiMondino20+} which synthesizes \cite[\S 7]{McCann20}:

\begin{definition}[Timelike $q$-dualizability]
\label{D:dualizability}
In a proper Lorentzian prelength space, we say $(\mu^-,\mu^+) \in \cP_c(M)^2$ is {\em timelike $q$-dualizable} (by $\gamma$) if (i) $\ell_q(\mu^-,\mu^+) \in (-\infty,0)$ and 
(ii) $\gamma \in \Gamma^q_\ll(\mu^-,\mu^+)$.
We say $(\mu^-,\mu^+)$ is strongly timelike $q$-dualizable if, in addition, (iii) there is a measurable $(-\ell)^q$-cyclically monotone set
$S \subset M^2_\ll \cap \spt[\mu^- \otimes \mu^+]$ such that $\gamma \in \Gamma_\le(\mu^-,\mu^+)$ is $\ell_q$-optimal 
if and only if $\gamma[S]=1$.  To be $(-\ell)^q$ cyclically monotone means every sequence $((x_i,y_i))_{i=1}^\infty$ in $S$
satisfies
$$
\sum_{i=1}^j (-\ell(x_i,y_i))^q \ge (-\ell(x_1,y_j))^q + \sum_{i=2}^j (-\ell(x_i,y_{i-1}))^q
$$
for each $j \in \bN$, with the convention $(-\infty)^q=-\infty$.
\end{definition}

Essentially,  timelike $q$-dualizability implies existence of a timelike optimizer (i.e. $\qopT(\mu^-,\mu^+)$ is non-empty), whereas strong timelike $q$-dualizability implies all optimizers are timelike: $\qopt \subset \Gamma_\ll(\mu^-,\mu^+)$.

We shall also need to consider probability measures on $\ell$-geodesics $\TGeol$, sometimes called {\em dynamic transport plans} or simply {\bf plans}.  
Recall that the space $C([0,1]; M)$ of continuous paths in $M$
with the uniform metric \eqref{uniform metric}
 is Polish if $(M,d)$ is.  When $M$ is a compact causal geodesic space, then $\TGeol \cap C([0,1];M)$ is $\sigma$-compact, since \cite[Corollary B.7]{Braun22.6+}
 shows $\{ \sigma \in \TGeol\cap C([0,1];M) \mid - \ell(\sigma(0),\sigma(1)) \ge r >0\}$ to be compact (and uniformly equicontinuous) for each $r>0$.
Denote by $\OptTGeoq(\mu,\nu)$ the set of measures $\eta \in \cP(\TGeol)$ for which $(e_0 \times e_1)_\#\eta \in \qopT(\mu,\nu)$:
these are examples of {\bf optimal plans};
here $e_t:\sigma \in C([0,1];M) \mapsto \sigma(t)$ denotes the time $t$ evaluation map.  
Braun's Proposition B.9 [ibid] uses measurable selection techniques to show that 
when $(\mu,\nu)$ is timelike $q$-dualizable (so that timelike optimizers exist), each 
$\gamma \in \qopT(\mu,\nu)$ is induced as 
$\gamma = (e_0 \times e_1)_\#\eta$ for some $\eta \in \OptTGeoq$.  As in Cavalletti \& Mondino's Proposition 2.32.7 \cite{CavallettiMondino20+}, $s \in [0,1] \mapsto \mu_s := (e_s)_\#\eta$ is then an $\ell_q$-path,  narrowly continuous (against bounded continuous test functions), 
and $d_1$-rectifiable, hence an $\ell_q$-geodesic in
$(\cP_c(M),d_1,\ell_q)$:

\begin{lemma}[On $\ell_q$-geodesics in $(\cP_c(M),d_1,\ell_q)$]
\label{L:d_1 rectifiability}
Fix $q \in (0,1]$ and let $(M,d,\ell)$ be a causal geodesic space and $\eta \in \OptTGeoq(\mu,\nu)$ for some $\mu,\nu \in \cP_c(M)$.  
Then $s \in [0,1] \mapsto \mu_s := (e_s)_\#\eta$ is an $\ell_q$-path, narrowly continuous, and $d_1$-rectifiable:
\end{lemma}

\begin{proof}
1. For $0\le s<t \le 1$ we have
\begin{align*}
|\ell_q(\mu_s,\mu_t)|^q 
&\ge \int_{\TGeol} |\ell(\sigma(s),\sigma(t))|^q d\eta(\sigma) 
\\&= (t-s)^q \int |\ell(x,y)|^q d(e_0,e_1)_\# \eta
\\&= (t-s)^q |\ell_q(\mu_0,\mu_1)|^q
\\&> 0
\end{align*}
by the definition of $\OptTGeoq(\mu_0,\mu_1)$.  
Using the $q$-th root of this to estimate
$$
\ell_q(\mu_0,\mu_1) \le \ell_q(\mu_0,\mu_s) +  \ell_q(\mu_s,\mu_t) + \ell_q(\mu_t,\mu_1)
$$ 
shows all these nonstrict inequalities must be saturated,  so $(\mu_s)_{s\in[0,1]}$ is an $\ell_q$-path.

2.  Letting $f \in C(M)$ be continuous and bounded yields
\begin{align*}
\lim_{s \to t} \int f d \mu_s
&=\lim_{s \to t} \int_{\TGeol} f(\sigma(s)) d\eta(\sigma) 
\\&=  \int f d\mu_t
\end{align*}
by the dominated convergence theorem and the continuity of $\ell$-paths (Lemma \ref{L:continuity of geodesic paths}).

3.  In a causal geodesic space,  compactness of $Z := J(X_0,X_1)$ follows from that of $X_i=\spt \mu_i$,
and there exists a bound $B$ on the $d$-length of causal curves in $Z$.
 Letting $0=t_0 < t_1 <\cdots < t_j =1$ be an arbitrary partition,
\begin{align*}
\sum_{i=1}^j d_1(\mu(t_{i}),\mu(t_{i-1})) 
&\le \sum_{i=1}^j \int_{\TGeol} d(\sigma(t_i),\sigma(t_{i-1})) d\eta(\sigma) 
\\&\le \int L_d(\sigma) d\eta(\sigma)
\\&\le B.
\end{align*}
Taking the supremum over partitions yields $L_{d_1}(\mu) \le B$ as desired.
\end{proof}
  
\begin{remark}[Relating timelike $q$-dualizability to $\ell_q$-geodesics]\label{R:existence of plans}
Fix a proper causal geodesic space $(M,d,\ell)$. If $\mu_0,\mu_1 \in \cP_c(M)$ are timelike $q$-dualizable --- so that
$\gamma \in \Gamma^q_\ll(\mu_0,\mu_1)$ exists --- we have just argued that there exists an $\ell_q$-geodesic $(\mu_t)_{t \in[0,1]}$ 
with $\mu_t = (e_t)_\#\eta$ and $\gamma=(e_0 \times e_1)_\#\eta$ for some $\eta \in \OptTGeoq(\mu_0,\mu_1)$.
Similarly, when $\Gamma^q(\mu_0,\mu_1) = \Gamma^q_\ll (\mu_0,\mu_1)$ for some $\mu_0,\mu_1 \in \cP_c(M)$---
as when $(\mu_0,\mu_1)$ are strongly timelike $q$-dualizable --- each $\gamma \in \Gamma^q_\ll(\mu_0,\mu_1)$ 
arises as $\gamma=(e_0 \times e_1)_\#\eta$ from some $\eta \in \OptTGeoq(\mu_0,\mu_1)$; this induces
a $d_1$-Lipschitz $\ell_q$-geodesic $\mu_t := (e_t)_\#\eta$ in $(\cP_c(M),d_1,\ell_q)$ by Lemma \ref{L:d_1 rectifiability}.  
In the strongly timelike $q$-dualizable case one might expect all $\ell_q$-paths to arise in this way,
but a proof of this conjecture remains elusive unless we assume compactness of the closed set $\AGeo^\ell(J(X,X))$ in the uniform
metric \eqref{uniform metric} for each compact subset $X \subset M$
(in which case a dyadic optimization can be iterated and a limiting measure $\eta$ on $\AGeol(J(\spt\mu_0,\spt\mu_1))$ extracted as in
Theorem 2.10 of \cite{AmbrosioGigli13o},
 with $\gamma=(e_0 \times e_1)_\#\eta \in \Gamma^q(\mu_0,\mu_1)$ optimizing, hence timelike).
\end{remark}

\subsection{Smooth versus metric-measure spacetimes}

The other major ingredient which goes into defining synthetic lower Ricci curvature bounds is a reference measure $m$ on $(M,d)$.
A metric spacetime $(M,d,\ell)$ equipped with a nonnegative Borel measure $m$ is therefore called a  {\bf metric-measure spacetime} or  {\em m.m.s.t}.
We call $(M,d,\ell,m)$ a {\bf measured causal geodesic space} if, in addition $(M,d,\ell)$ is a causal geodesic space and $m$ 
assigns finite mass to bounded sets. 

\begin{definition}[Smooth spacetimes]\label{D:smooth spacetime}
We call a smooth, connected, Hausdorff, time-oriented, $n$-dimensional Lorentzian manifold $(M^n,g)$ of signature $(+-\ldots-)$ a  {\em (smooth) spacetime}. 
\end{definition}

Results of Nomizu, Ozeki \cite{NomizuOzeki61} and Geroch \cite{Geroch68} 
show that any smooth spacetime is second countable and that its topology is metrized by the distance $d_{\tilde g}$ induced by some 
complete Riemannian metric $\tilde g$.  
The prototypical example of a proper causal geodesic space (and also of a proper globally hyperbolic regular LLS) therefore consists of a smooth spacetime endowed with its usual causality relations and time-separation function $\tau=\ell_-$ induced by $g$,  under the additional assumptions that there are no closed causal loops and for each $x,y \in M$ the causal diamond $J(x,y)$ is compact (corresponding to global hyperbolicity).
We call such an object a {\bf smooth globally hyperbolic spacetime}.

\begin{remark}[
Nonunique geodesics of measures on manifolds]\label{R:unique}
On a smooth globally hyperbolic spacetime $(M^n,g)$ with $0<q<1$,  timelike $q$-dualizability of $(\mu_0,\mu_1) \in \cP_c(M)^2$
and absolute continuity of $\mu_0$ or $\mu_1$ with respect to $\vg$ imply existence of a unique 
$\eta \in \cP(\TGeol)$ such that $\mu_t = (e_t)_\#\eta$ is an $\ell_q$-path from $\mu_0$ to $\mu_1$.
This is because strict convexity of the Lagrangian in the timelike region 
allows the duration and direction separating $\mu_0$-a.e. event from its partner to be extracted from derivatives of the solution
of the dual linear program to \eqref{l_q} \cite[Theorem 7.1]{McCann20}.  
Taking derivatives poses obvious challenges at the light cone and in the nonsmooth setting \cite{BeranOctet23+}.
When timelike $q$-dualizability fails to be strong, 
uniqueness {may fail the broader class of measures $\cP\left(\AGeol \right)$ on causal (rather than timelike) $\ell$-geodesics,
as the following example illustrates.}
\end{remark}

\begin{example}[Mixed timelike and lightlike translations]
Fix the usual coordinates $(x^1,x^2)$ on the plane $M=\R^2$ and consider the Minkowski metric $g_{ij} dx^i dx^j = dx_1^2 - dx_2^2$,  Euclidean metric
$\tilde g_{ij} dx^i dx^j = dx_1^2 + dx_2^2$, and
associated time-separation function $\ell=\ell_g$, distance $d=d_{\tilde g}$,
and area element 
$\vg$.  
For $T > 0$,   
consider $\mu \in \cP_c(M)$ and its time translate $\nu=F^T_\#\mu$ by $F^T(x^1,x^2) = (x^1+T,x^2)$.
Using $\gamma:=(id \times F^T)_\#\mu$ as a trial measure in \eqref{l_q} gives $\ell_q(\mu,\nu) \le -T$.
Moreover, we claim that equality holds, meaning $F^T$ is an $\ell_q$-optimal map; 
when $T$ is sufficiently large (so $\mu$ and $\nu$ are $q$-separated),  
this claim follows from \cite[Theorem 5.9]{McCann20} by constructing explicit dual potentials $u(x) + v(y) \ge (-\ell(x,y))^q$ on
$\spt [\mu \times \nu]$ which produce equality $\gamma$-a.e.  For smaller $T$ it then follows from Corollary 5.9 of the same reference.
Thus translation by any timelike future-directed vector in Minkowski space is $\ell_q$-optimal.  
To show translation by any future-directed null vector is also $\ell_q$-optimal requires an additional argument.
Let $G^T(x^1,x^2):=(x^1+T,x^2+T)$ and $\omega:=G^T_\#\mu$.  Using the trial measure $\gamma:=(id \times G^T)_\#\mu$
shows $\ell_q(\mu,\omega)\le 0$; again we claim that equality holds.  To verify this, consider the family of Minkowski metrics $g^\epsilon= (g + \epsilon \tilde g)$
having wider light cones but converging to $g=g^0$ as $\epsilon \searrow 0$, and their time-separation functions $\ell^\epsilon \le \ell$.  Notice $G_T$ represents a timelike translation hence
$\ell^\epsilon_q$-optimal map for all $\epsilon>0$. Thus $\ell_q^\epsilon(\mu,\omega) = T\sqrt{2 \epsilon} \le \ell_q(\mu,\omega)\le 0$ tends to zero as $\epsilon \searrow 0$,
hence $\ell_q(\mu,\omega)=0$ as desired.
For $T>0$ large enough that $\mu$ and $\omega$ have disjoint support,  there will be also be $\ell_q$-optimal maps other than $G^T$, including one which is order-reversing
instead of order-preserving along the (right-moving) lightlight geodesics in $M=\R^2$.  This is analogous to the better known nonuniqueness of optimal measures $\gamma$ 
attaining $d_1(\mu,\omega)$ 
(which cannot be resolved without adding some requirement of monotonicity in the direction of transport \cite{FeldmanMcCann02u}).  Now imagine $\mu = 1_{B_r(-z)} + 1_{B_r(z)}$ to consist of a uniform measure on the disjoint union of two far apart Euclidean balls centered at $\pm z=(0,\pm R)$ with $r<1<R$. For $T>2r$ we find the optimal measure
between $\mu$ and $\nu =  F^T_\#(1_{B_r(-z)}) + G^T_\#(1_{B_r(z)})$ is non-unique,  consisting of the timelike translation $F^T$ on the first ball,  and either the lightlike translation $G^T$ or any of the other optimal options previously asserted to exist on the second ball.
\end{example}

\begin{definition}[
Bakry-\'Emery, smooth m.m.s.t., $N$-Ricci curvature
]\label{D:Bakry-Emery spacetime}
When equipped with its Lorentian volume $dm=d\vg$ a smooth globally hyperbolic spacetime 
becomes a measured causal geodesic space.  We may instead choose to equip $(M^n,g)$
with a volume given by a smooth weight $V \in C^\infty(M)$
$$
dm(x) = e^{-V(x)}d\vg(x),
$$
in which case it becomes an example of a {\bf smooth metric-measure spacetime},  also known as a Bakry-\'Emery spacetime (in honor of \cite{BakryEmery85}).
Associated to the weight and a parameter $N\in (-\infty,\infty]$ is a modification of the Ricci tensor,  known as the {\bf $N$-Ricci} or {\bf Bakry-\'Emery} tensor
\cite{Case10} 
\begin{equation}\label{N-Ricci}
\Rc^{(N,V)} := \Rc + D^2 V - \frac1{N-n} DV \otimes DV;
\end{equation}
we require $V=const$ if $N=n$.
\end{definition}

\subsection{Synthetic timelike convergence conditions}

Let $\mu\ge 0$ be a Borel measure on a metric-measure spacetime $(M, d,\ell, m)$,
and assume $\mu$ vanishes outside a set $S$ of finite $m$-volume.  If $\mu$ is absolutely continuous
with respect to $m$, we define its Boltzmann-Shannon relative {\bf entropy} by 
\begin{equation}\label{entropy}
H(\mu \mid m) := 
\int_M \frac{d\mu}{dm} \log \frac{d\mu}{dm} dm,  
\end{equation}
setting $H(\mu \mid m) =+\infty$ otherwise.  Notice $H(\mu \mid m) > \mu(M)-m(S)$ has a well-defined value in $(-\infty,\infty]$ since $s \log s \ge s-1$.

For $K \in \R$ and $N>0$,  recall a function $h:\R \longrightarrow \R \cup \{+\infty\}$ is said to be {\bf $(K,N)$-convex} 
if: $h$ is upper semicontinuous, $\Dom h:= \{s \in \R \mid h(s)<\infty\}$ is connected,  $h$ is {\em semiconvex} (meaning
$s \in \R \mapsto h(s) + \lambda s^2$ is convex for $\lambda>1$ sufficiently large), and
$$
h''(s) - \frac1N (h'(s))^2 \ge K
$$
holds throughout the interior of $\Dom h$ in the distributional sense \cite{ErbarKuwadaSturm15}.

On a smooth globally hyperbolic spacetime one obtains my characterization \cite{McCann20}
of Case \cite{Case10}, Woolgar \& Wylie's timelike lower $N$-Ricci curvature bounds  \cite{WoolgarWylie16}; 
see also Mondino \& Suhr \cite{MondinoSuhr23}, Cavalletti \& Mondino \cite
{CavallettiMondino20+}, and the proof at the end of this section addressing the case $0<N<n$.

\begin{theorem}[Timelike curvature-dimension bounds]\label{T:consistency}
Let $(M^n,g)$ be a smooth globally hyperbolic spacetime, $0<q<1$, $K \in \R$ and let $\vg$ denote the Lorentzian volume measure.
Fix $N \in (0,\infty]$ and $dm = e^{-V}d\vg$ for some $V \in C^\infty(M)$.
Suppose (i) for any strongly timelike $q$-dualizable $(\mu_0,\mu_1) \in \cP_c(M)^2$ having finite entropy,  there is a measure
$\eta \in \cP(\TGeol)$ such that $\mu_t = (e_t)_\#\eta$ is an $\ell_q$-path from $\mu_0$ to $\mu_1$ 
along which $H(\mu_t \mid m)$
is a $(KT^2,N)$-convex function of $t \in [0,1]$, with $T=\|L_\ell\|_{L^2(\eta)}$ as in \eqref{Lorentzian length}.
Then (ii) the same statement holds if we drop the adjective `strongly'.  
Moreover, either of these statements is equivalent to
(iii) $N\ge n$ and 
$\Rc^{(N,V)}(v,v) \ge Kg(v,v)$ for every timelike $v \in TM$.
\end{theorem}

Unlike property (iii), properties (i) and (ii) of Theorem \ref{T:consistency} make sense in a metric-measure spacetime independently
of whether or not it possesses any manifold structure.  Motivated by this characterization (and by developments in positive signature, such as  
\cite{LottVillani09} \cite{Sturm06a} \cite{Sturm06b} \cite{ErbarKuwadaSturm15}),  
Cavalletti \& Mondino defined a family of timelike curvature dimension condition as follows: 
given parameters $(K,N,q) \in \R \times (0,\infty) \times (0,1)$, 
a proper measured causal geodesic space $(M,d,\ell,m)$ satisfies $\TCD^e_q(K,N)$
if and only if property (ii) of Theorem \ref{T:consistency} holds; it satisfies a weaker variant $w\TCD^e_q(K,N)$ if and only if (i) of the same theorem holds.
Here $\TCD$ stands for {\bf timelike curvature dimension} with the superscript $e$ denoting the {\bf entropic} variant defined using the (logarithmic)
Boltzmann-Shannon entropy \eqref{entropy} as in Erbar, Kuwada \& Sturm \cite{ErbarKuwadaSturm15} and McCann \cite{McCann20} instead of the $N$-R\'enyi (power law) entropy 
used by Lott \& Villani, \cite{LottVillani09}, Sturm \cite{Sturm06b}, Braun \cite{Braun22.6+}, etc.
They go on to show such spaces have many remarkable properties, such as timelike Bishop-Gromov and Brunn-Minkowski inequalities.  
Moreover,  they show a version of the Hawking singularity theorem
 remains true in $w\TCD^e_q(0,N)$ spaces. 
  Simultaneously and independently,  a nonsmooth analog of Hawking's theorem was proven in positive signature
 by Burtscher, Ketterer, Woolgar and myself \cite{BurtscherKettererMcCannWoolgar20}: in a $\CD(K,N)$ space,  we show that any mean convex set obeys
 an explicit bound on its inscribed radius, and  in $\RCD(K,N)$ spaces we are able to classify the cases of equality.

For Lorentzian manifolds,  convergence notions analogous to Gromov-Hausdorff convergence have been considered since works of 
Noldus \cite{Noldus02} \cite{Noldus04a} \cite{Noldus04b} and Bombelli \cite{BombelliNoldus04}.  Sormani and Vega \cite{SormaniVega16} showed how a time function on a smooth spacetime generates a (so-called {\em null}) distance which topologizes the manifold so that Gromov-Hausdorff notions of convergence of spaces can then be applied.
This approach has been further developed by Allen and Burtscher \cite{AllenBurtscher22}, and extended to Lorentzian length spaces by Kunzinger \& Steinbauer
\cite{KunzingerSteinbauer22}, who explored its compatibility with sectional curvature bounds. These notions of convergence depend on the Lorentzian metric and time function alone,
and not on a reference measure.  In contrast,
Cavalletti \& Mondino introduced a notion of {\bf pointed-measured convergence} for metric-measure spacetimes, and show the set $w\TCD^e_q(K,N)$ contains all such limits of $\TCD^e_q(K,N)$ spaces \cite{CavallettiMondino20+}.  More precisely,  they say a sequence $(M_j,d_j,\ell_j,m_j)$ of measured causally
geodesic spaces with $x_j \in M_j$ for all $j \in \bN \cap \{\infty\}$ converges to a limiting measured causal geodesic space $(M_\infty,d_\infty,\ell_\infty,m_\infty,x_\infty)$ 
 if and only if the entire sequence embeds $d$-continuously
and $\ell$-isometrically into a proper causal geodesic space $(M,d,\ell)$,  so that $d(x_j,x_\infty)\to 0$ and $m_j \to m_\infty$ weakly against continuous compactly supported test functions.
This is modelled on one out of several equivalent notions of convergence in positive signature \cite{GigliMondinoSavare15}; in a doubling space all 
are equivalent to Gromov-Hausdorff convergence.
Distinct but related variants of Gromov-Hausdorff convergence have been proposed for bounded Lorentzian metric spaces by Minguzzi \& Suhr \cite{MinguzziSuhr22+} 
and using a functorial approach by Mueller \cite{Mueller22+}; however they involve synthesizing nonsmooth Lorentzian geometry in different ways.
Minguzzi \& Suhr also show that their notion of convergence preserves sectional curvature bounds formulated through triangle comparison.

Cavalletti \& Mondino do not show either $\TCD_q^e(K,N)$ or $w\TCD_q^e(K,N)$ to be closed for their notion of convergence.
Rather, given a pair of timelike $q$-dualizable measures $(\mu_0^\infty,\mu_1^\infty)$ in the limiting space,  they extract the required optimal plan $\eta^\infty\in \cP(\AGeo^{\ell_\infty})$ as a narrow limit of optimal plans $\eta^j \in \cP(\TGeo^{\ell_j})$ of timelike $q$-dualizable pairs $(\mu_0^j,\mu_1^j)$ in the approximating sequence of spaces.  For $\eta$ to vanish
outside $\TGeo^{\ell_\infty}$ as desired requires {\em strong} timelike $q$-dualizability of $(\mu_0^\infty,\mu_1^\infty)$;  while such a limit can be approximated by 
 timelike $q$-dualizable pairs $(\mu_0^j,\mu_1^j)$,  it is not clear that the  timelike $q$-dualizable of these pairs can be taken to be strong. 
 For measured causal geodesic spaces which are 
{timelike non-branching} however,  Braun showed that, as in the smooth case \cite{McCann20}, 
the weak and strong variations of the $\TCD^e_q$ condition coincide \cite{Braun22.6+}.
In fact,  Braun goes further by showing the strong and weak variations coincide on a somewhat larger collection of spaces which do not branch too much,  as quantified by his {\bf $q$-essentially timelike nonbranching condition} introduced in analogy with Rajala \& Sturm's \cite{RajalaSturm14}.  Under this condition,  he also shows the Boltzmann entropy can be replaced by the {\em $N$-R\'enyi entropy} in the definitions of $\TCD_q^e$,  provided the notion of $(K,N)$-convexity is modified using appropriate distortion coefficients.  In contrast to the $\RCD$ condition from the positive signature theory,   
 no stable variant of the $\TCD$ condition has yet been identified that implies the space is
 timelike nonbranching \cite{DengPhD21} or even $q$-essentially timelike nonbranching \cite{RajalaSturm14} as in \cite{AmbrosioGigliSavare14d}  \cite{GigliRajalaSturm16}.
 
 Since existing literature does not seem to address the case $0<N<n$ of Theorem \ref{T:consistency}, we close this section with a proof which covers that case while illustrating the power of the above-mentioned techniques.

\medskip

\begin{proof}[Proof of Theorem \ref{T:consistency}]
Clearly (ii) implies (i), so let us turn to the implication (iii) $\implies$ (ii). 

Assume (iii) holds; then $N \ge n$. If $N > n$, the implication (iii) $\implies$ (ii) follows from Corollary 7.5 of \cite{McCann20}.  On the other hand, if $N=n$ then $V=const$ by hypothesis;  in this case $\Rc^{(n,V)} = \Rc^{(n+\epsilon,V)}$ so applying the preceding sentence
with $N=n+\epsilon$ and then taking the limit $\epsilon \searrow 0$ yields (ii) as desired.

Finally,  we obtain the implication (i) $\implies$ (iii) by considering the two ways in which (iii) can fail.  
If $\Rc^{(N,V)}(v,v) < Kg(v,v)$ for some timelike $v \in TM$,  Theorem 8.5 of \cite{McCann20} implies (i) cannot hold.
It remains to show $0<N<n$ also causes (i) to fail.  A direct proof could be constructed by interpolating 
between any absolutely continuous measure $\mu_0$ supported in the timelike past of a Dirac measure $\mu_1$ 
using techniques of \cite{McCann20},
but to do so would be tedious except in the special case $X_\infty=(\R^n_1,\epsilon,\vol_\epsilon)$ of $\R^n$ equipped with its
usual Minkowski metric $\epsilon$ and volume.
In that special case, one finds the relative entropy along the $\ell_q$-path $\mu_s= (1-s)_\#\mu_0$ takes the form $h(s) = H(\mu_s \mid \vol_\epsilon) = h(0)-n \log (1-s)$, so if
$0<N<n$ then
$$
h''(s) - \frac{h'(s)^2}{N} = \frac{n}{(1-s)^2}(1 -\frac{n}{N}) 
$$
diverges to $-\infty$ as $s \to 1$.  This shows $X_\infty \not \in w\TCD(0,N)$ for all $0<N<n$.
On the other hand,  for any fixed point the dilations $X_\lambda:= (M^n,\lambda^2 g,\lambda^n e^{-V}d\vol_g)$
of a smooth metric-measure spacetime around a fixed point $\bar x \in M^n$ converge to Minkowski space $X_\infty$ as $\lambda \to \infty$ in the pointed measured sense of Cavalletti \& Mondino
(with the additive normalization $V(\bar x) = 0$).  If $X_1 \in \TCD_q^e(K,N)$ satisfies (i) for some $N>0$ and $K \in \R$,  then $X_\lambda \in \TCD_q^e(K/\lambda^2,N)$ and the stability result \cite{CavallettiMondino20+} described above yields $X_\infty \in w\TCD_q^e(0,N)$.
But this produces the desired contradiction to  $0<N < n$.
\end{proof}
 
\section{A synthetic null energy condition}
\label{S:NEC}

A successful non-smooth theory of curvature bounds should have three properties:
(a) consistency (b) stability and (c) consequences.   
{\em Consistency} means that it should reduce to the classical notion in the smooth setting.  
{\em Stability} means it should be preserved under suitable limits.  {\em Consequences}
means it should have interesting implications.  Like its progenitors in positive signature
\cite{LottVillani09} \cite{Sturm06a} \cite{Sturm06b},  Cavalletti \& Mondino have shown their timelike 
curvature dimension conditions have versions of all three.  

Let us now turn to an open question highlighted in 
\cite{CavallettiMondino22}:  to find a nonsmooth version of the {\bf null convergence condition} (\NC):
\begin{equation}\label{NC}
\Rc(v,v) \ge 0 \quad {\rm whenever}\ g(v,v)=0. 
\end{equation}
Since $G(v,v)=\Rc(v,v)$ for all null vectors $v$ (irrespective of cosmological constant),
if Einstein's field equation $G= 8\pi T$ holds then
condition \eqref{NC} becomes equivalent to the {\bf null energy condition} (\NE),
\begin{equation}
T(v,v) \ge 0 \quad {\rm whenever}\ g(v,v)=0.
\end{equation}

Although it is tempting to try to answer this question by developing a theory of measure transportation along the lightlike geodesics from Remark \ref{R:AGeol}
(analogous to the timelike transport theory of the previous section),  many technical challenges arise.  
Somewhat surprisingly, we are able to cirmcumvent these difficulties using a different approach, based on the following smooth theorem that we now establish
in general pseudo-Riemannian signature:
\begin{theorem}
[Null bounds imply non-null bounds locally]  \label{T:NECboundsT}
Let $(M^n,g_{ij})$ be a smooth 
pseudo-Riemannian manifold with a continuously differentiable tensor field $F_{ij}$ satisfying
$F(v,v) \ge 0$ for all null vectors $v$.  
Then for each compact subset $Z \subset M$  
there is a constant $C_Z \in \R$ such that $F (p,p) \ge C_Z |g(p,p)|$ for all $p \in T_zM$ with $z \in Z$.
\end{theorem}

\begin{proof}
Without loss of generality, take $F$ to be symmetric and assume $g$ is not Riemannian (since the Riemannian case is standard). 
Then $(M^n,g_{ij})$ also admits a complete Riemannian metric $\tilde g$ by results of Nomizu \& Ozeki \cite{NomizuOzeki61} and Geroch \cite{Geroch68}. 
Let $S$ denote the subset $(p,z) \in TM$ of the sphere bundle satisfying $\tilde g(p,p)=1$,
and decompose $S=S_+ \cup S_- \cup S_0$ into timelike, spacelike and null-vectors. 
By $2$-homogeneity, it is enough to establish the desired bound for all $(p,z) \in S \setminus S_0$ with $z \in Z$.  
Choose an infimizing sequence $(p_j,z_j) \in S\setminus S_0$
for the ratio 
$$
r(p,z) := \frac{F_z(p,p)}{|g_z(p,p)|}
$$
subject to ($z_j$ and) $z \in Z$.  We claim the limit
$$
C_Z := \lim_{j \to \infty} r(p_j,z_j) =  \inf_{(p,z) \in S\setminus S_0} r(p,z)
$$
is finite.  To derive a contradiction,  assume $C_Z=-\infty$.
A  (nonrelabelled) subsequence $(p_j,z_j)$ converges to a limit $(v_\infty,z_\infty)$ by the compactness of $Z$ and the fibres of $S$.
Since $r(p,z)$ takes continuous real values on $S\setminus S_0$,  our assumption $C_Z=-\infty$ implies the denominator must  become vanishly small, hence 
$v_\infty \in S_0$ (i.e. is null).
To estimate the ratio $r(p,z)$,  we shall Taylor expand its numerator and denominator, after first
 choosing smooth coordinates $(x_1,\ldots,x_{2n-1})=(X,x_{2n-1})$ on the sphere bundle $S$ in which $(v_\infty,z_\infty)$ becomes the origin $x=0$ and $S_0$ becomes the boundary of the halfspace $\R^{2n-1}_+$ locally. 
 These coordinates exist because transversality of the intersection of the sphere bundle with the null bundle guarantees $S_0$ is a smooth submanifold of $S$  \cite[\S 1.5]{GuilleminPollack74},
so  any local choice of coordinates on $S_0$ may be extended to the desired coordinates on $S$ using the signed distance  to $S_0$ (induced by restricting the Sasakian extension of
$\tilde g$ to $S\subset TM$).
 To avoid interrupting the flow of ideas,  we postpone verifying the claimed transversality to the end of the proof.

 The metric $g_z(p,p)$ takes opposite signs on $S_+$ and $S_-$,   hence only the last of its $2n-1$ partial derivatives can be non-vanishing on $\p \R^{2n-1}_+$; it is 
 strictly non-vanishing by the non-degeneracy of $g$. 
The numerator $F_z(p,p)$ 
is non-negative on $\p \R^{2n-1}_+$ by hypothesis. 
Taylor expansion on $S \setminus S_0$ around $S_0$ therefore shows
$$
r (x_1,\ldots,x_{2n-1}) =\frac{F(x)}{|g(x)|}\ge 
\frac{x_{2n-1} \frac{\p F}{\p x_{2n-1}} (X,0) + o(|x_{2n-1}|)}{\left| x_{2n-1} \frac{\p g}{\p x_{2n-1}}(X,0) + o(|x_{2n-1}|)\right|}   
$$
tends to a limit $C_Z \ge -|\frac{\p F / \p x_{2n-1}}{\p g /\p x_{2n-1}}(0)|$ as $x=(X,x_{2n-1}) \in \R^{2n-1}\setminus\{x_{2n-1}=0\}$ 
tends to the origin along the original minimizing sequence.  This contradicts 
$C_Z = -\infty$.  

It remains only to verify the claimed transversality of $S = \{\tilde g(p,p)=1\}$ and $L=\{ g(v,v)=0\} \setminus 0_M$, where $0_M$ is the zero section of $TM$.
Since $\tilde g$ and $g$ are both nondegenerate, the implicit function theorem shows $S$ and $L$ are both smooth hypersurfaces in $TM$.  To show they intersect transversally, it
is therefore sufficient to show at each point  $(v,z) \in S_0 = S \cap L$,  the $2n-1$ dimensional tangent spaces $S_{(v,z)}$ (to $S$) and $L_{(v,z)}$ (to $L$) do not coincide, for then their sum must have full dimension $2n$.
Thus it is enough to show that the $n-1$ dimensional tangent spaces $S_z$ to the ellipsoid $S \cap T_zM$ and $L_z$ to the nullcone $L \cap T_zM$ are distinct.
But this follows from the fact that $v+S_z$ lies outside the ellipsoid, hence contains no ray through the origin of $T_z M$,  whereas $L_z$ contains the lightlike ray $\{\lambda v  \mid \lambda \in \R\}$ through the origin and $v$.
 \end{proof}

\begin{corollary}[(\NC) 
versus variable timelike lower Ricci bounds]\label{C:NECboundsT}  
Let $(M^n,g_{ij})$ be a smooth pseudo-Riemannian 
manifold with a continuously differentiable tensor field $F_{ij}$.
Then $F(v,v) \ge 0$ for all null vectors $(v,z) \in TM$ if and only if:
for each compact subset $Z \subset M$  
$$
C_Z^T := \inf_{z \in Z} \inf_{g_z(p,p)>0} \frac{F(p,p)}{g(p,p)} > -\infty
$$
\end{corollary}

\begin{proof}
One direction follows directly from Theorem \ref{T:NECboundsT}.
We'll show the contrapositive of the other:  if $C^T_{\{z\}}$ is finite then $F(p,p) \ge C^T_{\{z\}} g (p,p)$ for all $p \in T_z M$ with $g_z(p,p)>0$. 
The same extends to null $v$ by continuity,   so $C^T_Z>-\infty$ implies $F(v,v) \ge 0$ for all null $v$, as desired.
\end{proof}

\begin{remark}[Weighted null energy vs weighted Ricci bounds]\label{R:NECboundsT}
On a smooth Lorentzian manifold $(M^n,g_{ij})$ with weighted volume $dm(x)=e^{-V(x)}d\vg(x)$ and $N\ne n$,
the previous corollary applied either to $F=\Rc$ or 
$$F=\Rc^{(N,V)} = \Rc + D^2 V - \frac{1}{N-n} DV \otimes DV
$$
shows non-negativity of the (weighted) Ricci tensor in null directions is equivalent to a local lower bound on the (weighted) Ricci 
curvature in timelike directions.  Alternately, by applying the corollary to the stress-energy tensor $F=T$, we see that the null
energy condition (NE) is equivalent to 
a variable lower bound on the stress-energy in timelike directions.
\end{remark}

Motivated by this equivalence in the smooth setting, we can define the null energy condition in a metric-measure spacetime as follows:

\begin{definition}[Synthetic null energy conditions]\label{D:NECq}
Given $N>0$ and $0<q<1$,  a proper measured causally geodesic space $(M,d,\ell,m)$ satisfies $\NC^e_q(N)$ if and only if for each compact $Z \subset M$,
there exists $K_Z \in \R$ such that $J(Z,Z)$ satisfies $\TCD^e_q(K_Z,N)$.    Similarly $(M,d,\ell,m) \in w\NC^e_q(N)$
if and only each compact $Z \subset M$ there exists $K_Z\in \R$ such that $J(Z,Z)$ satisfies $w\TCD^e_q(K_Z,N)$.
\end{definition}

\begin{remark}[Alternate definitions and equivalences]  
One can also define $(M,d,\ell,m) \in (w)\NC_q^{(*)}(N)$ if and only if for each compact 
$Z \subset M$, there exists $K_Z \in \R$ such that $J(Z,Z) \in (w)\TCD_q^{(*)}(K_Z,N)$.  
Here  the {\em reduced} timelike curvature dimension condition $\TCD^*_q(K,N)$ and
$\TCD_q(K,N) \subset \TCD_q^*(K,N)$ are defined by slightly different distorted convexity 
requirements of the $N$-R\'enyi entropy along $\ell_q$-geodesics \cite{Braun22.6+}.  Experience with positive signature \cite{Sturm06a} \cite{Sturm06b} \cite{BacherSturm10}
\cite{ErbarKuwadaSturm15} suggests
the more restrictive variant, although harder to work with,
is required to obtain sharp constants in geometric inequalities such as the time-to-singularity
in the Hawking theorem.  When $(M,d,\ell,m)$ is ($q$-essentially) timelike nonbranching,  
Braun \cite{Braun22.6+} establishes the equivalence of the reduced and entropic variants to each other and to
their weak versions ($w\TCD^e_q = \TCD^e_q = \TCD^*=w\TCD^*$) 
and hence the corresponding variations of $\NC_q$.  The case $N=\infty$ is discussed in~\cite{Braun22.5+}.
The equivalence of the above conditions to the sharp variant $\TCD_q$ (and independence of these notions on $q <1$) 
remain to be shown, though their positive signature analogs are known \cite{CavallettiMilman21} \cite{AkdemirCavallettiColinetMcCannSantarcangelo21}.
\end{remark}

Informally,  we say a metric-measure spacetime satisfies the null convergence condition if and only if it satisfies timelike lower Ricci curvature bounds locally,
or equivalently,  if and only if it satisfies a variable lower bound on its timelike Ricci curvature.  This already sheds much light on our desiderata
(a) consistency (b) stability and (c) consequences.  Consistency is almost for free:

 \begin{theorem}[Consistency with the null energy condition]\label{T:NECconsistency}
Let $(M^n,g)$ be a smooth globally hyperbolic spacetime, $0<q<1$, and let $\ell$ and $\vg$ denote the Lorentzian time separation function
and volume measure.
Fix $N \in (0,\infty]$ and $dm = e^{-V}d\vg$ for some $V \in C^\infty(M)$, and a complete auxiliary Riemannian metric $\tilde g$ on $M^n$
which induces a distance $d$.  Then the following are equivalent: (i) $(M,d,\ell,m) \in w\NC^e_q(N)$
(ii) $(M,d,\ell,m) \in \NC^e_q(N)$ and (iii) $N \ge n$ (with $V=const$ as usual in case $N=n$) and every null vector $(v,z) \in TM$ satisfies
$$
\Rc^{(N,V)}(v,v) \ge 0.
$$
\end{theorem}

\begin{proof}
Clearly (ii) implies (i).  
To show (i) implies (iii),  assume (i) holds, so that our smooth globally hyperbolic spacetime satisfies 
$(M,d_{\tilde g},\ell_g,dm=e^{-V}d\vg) \in w\NC_q^e(N)$. Recall a set $Z \subset M$ is called {\em causally convex} if $J(x,y) \subset Z$ for each $x,y \in Z$.
If $y \in M$ there exists a timelike diamond $I(x,z)$ containing $y$.  
Since the causal diamond $J(x,z)$ is compact and causally convex, we
have $J(x,z) \in w\TCD_q^e(K,N)$ for some $K \in \R$ by Definition \ref{D:NECq}. 
Its interior $I(x,z)$ inherits the property of being a smooth globally hyperbolic spacetime from $(M^n,g)$.
And the time-separation function induced on $I(x,z)$ by $g$ coincides with the restriction of $\ell_g$, by global hyperbolicity and the causal convexity of $I(x,z)$.
Thus Theorem \ref{T:consistency} implies $\Rc^{N,V}(p,p) \ge K g(p,p)$ for all timelike hence all causal $p \in T_yM$,  and (iii) follows by arbitrariness of $y \in M$.

We turn to the final implication, (iii) implies (ii).  Assume (iii) holds,  so that $N \ge n$, $V=const$ if $n=N$, and $\Rc^{N,V}(v,v) \ge 0$ holds for all null $(v,y) \in TM$. 
Fix a compact set $Z \subset M$. 
The future and past of $Z$ have Lipschitz boundaries according to Hawking \& Ellis \cite[Proposition 6.3.1]{HawkingEllis73}; these cover the boundary
of $J(Z,Z)$.  Let $X:=I(Z,Z)$ denote the interior of $J(Z,Z)$.
Then $X$ is causally convex, smooth and globally hyperbolic.
Moreover Theorems \ref{T:NECboundsT}, Corollary \ref{C:NECboundsT} and Remark \ref{R:NECboundsT} yield $K_Z\in \R$ such that 
$\Rc^{N,V}(p,p) \ge K_Z g(p,p)$ holds for all timelike $(p,y) \in TX$.  
Observe $Y := J(Z,Z)$ is a causal geodesic space which differs from $X$ only by the
aforementioned Lipschitz hypersurfaces.  Since these hypersurfaces have $\vg$ measure zero,  any $\mu_\pm \in \cP_c(Y)$ with finite entropy relative to $dm= e^{-V}d\vg$ must be absolutely continuous relative to $\vg$, hence restrict to a probability measure on $X$. If they were compactly supported in $X$ we could apply Theorem 
\ref{T:consistency} to conclude, but since that need not be the case  
assume $(\mu_-,\mu_+)  \in \cP_c(Y)^2$ are timelike dualizable and have finite entropy.
Remark \ref{R:existence of plans} yields
an $\eta \in \cP(\TGeol)$ generating a timelike $\ell_q$-path $\mu_t=(e_t)_\#\eta$ from $\mu_-$
to $\mu_+$.  By causal convexity, the measures $\mu_t$ all vanish outside of $X$ which lies in the compact set $Y$,  hence their entropies enjoy a uniform lower bound
$h(t) = H(\mu_t \mid m) \ge 1 - m[Y]>-\infty$.   
The $(K_Z T^2, N)$-convexity of $h$ now follows from Theorem 7.4 of \cite{McCann20}
as in Corollaries 7.5, 6.6 of the same reference, with Remark 6.7 there yielding $T:= \|L_\ell\|_{L^2(\eta)}$. This shows 
$Y = J(Z,Z) \in \TCD_q^e(K_Z,N)$.  Arbitrariness of the compact set $Z \subset M$ concludes the proof of (ii): $M \in \NC_q^e(N)$.
\end{proof}

\begin{remark}[Relaxing global hyperbolicity]
Global hyperbolicity for metric spacetimes plays a role analogous to bounded compactness (i.e. properness) for metric spaces.
Since the theory of lower Ricci curvature bounds can be developed for Polish rather than proper metric spaces (with properness 
following if $N<\infty$ \cite{Sturm06b}),  it is natural to 
expect that it may be possible to replace the global hyperbolicity assumed in the foregoing theorems by
an appropriate notion of completeness for metric spacetimes~\cite{BeranOctet23+}. 
\end{remark}

Taking $V=const$ in \eqref{N-Ricci} yields the null energy condition \eqref{NC} required for the Penrose singularity theorem \cite{HawkingEllis73} \cite{Penrose65a}.  On compact Lorentzian geodesic subsets of $\NC_q^e(N)$ spaces,
consequences follow from those established for $\TCD^e_q(K,N)$ spaces by Cavalletti, Mondino \cite{CavallettiMondino20+} and Braun \cite{Braun22.5+} \cite{Braun22.6+}.  
In positive signature,  variable lower Ricci curvature bounds have been studied by Sturm \cite{Sturm15} \cite{Sturm20}, Ketterer \cite{Ketterer15+} \cite{Ketterer17}, and  Braun, Habermann \& Sturm \cite{BraunHabermannSturm21}.  With Braun, we pursue the development of an analogous theory in Lorentzian signature \cite{BraunMcCann23p}.

On the other hand,  stability cannot hold in general:  we cannot expect the pointed measured limit $(M,d,\ell,m,x)_\infty$ of a  
sequence $(M,d,\ell,m,x)_j \subset \NC^e_q(N)$  in the sense of Cavalletti \& Mondino to lie in $w\NC^e_q(N)$
(unless one is willing to impose some independence of $j$ on the local timelike lower bounds $K_j=K_{Z_j}$ along the sequence).
The following example describes a sequence of smooth metric-measure spacetimes satisfying $\NC^e_q(\infty)$ which converge to a limit that is not in $w\NC^e_q(\infty)$.
While $\spt m_j$ is connected along the sequence,  $\spt m_\infty$ consists of two isolated points, lying on the future and past boundaries of the space.  
Although weights may seem exotic, they appear in physical contexts including Brans-Dicke theory \cite{Woolgar13} and the near horizon geometry of black holes
\cite{GallowayKhuriWoolgar21}. Several possible interpretations of our example seem possible.  On the one hand,  it suggests that severe topogical change may
also occur in the limit of spaces in $\NC_q^e(\infty)$;  this particular change cannot occur in a limit of an analogous sequence $(M,d,\ell,m,\bar x)_j \in \TCD_q^e(K,N)$, because  stability  \cite{CavallettiMondino20+} \cite{Braun22.5+} of the timelike curvature condition implies $\spt \mu_\infty$ is a Lorentzian geodesic space.
 On the other hand, the instability we demonstrate may simply indicate that the combination of $\NC_q^e(\infty)$ with the pointed measured topology do not lead to 
 a physically meaningful notion,  since arbitrarily small perturbations can push $(M,d,\ell,m,x)_\infty$ into $\NC_q^e(\infty)$.
 Perhaps a nonsmooth version of the {\em weak energy condition} 
 \begin{equation}\label{WEC}
T(v,v) \ge 0 \quad \forall g(v,v) \ge 0,
 \end{equation}
 or the {\em dominant energy condition,} which augments \eqref{WEC} by also requiring
 \begin{equation}\label{DEC}
g(Tv,Tv) \ge 0 \quad \forall g(v,v) \ge 0,
 \end{equation}
 would be more amenable to stability,  since the set of timelike directions forms an open set;
 here $8\pi T_{ab} =R_{ab} - \frac12 R g_{ab} -K g_{ab}$ by the Einstein field equation, and $K\in \R$ is the cosmological constant. 
Both \eqref{WEC}--\eqref{DEC} are expected to be satisfied by ordinary (but not quantum) matter,  and their combination
is known to prevent information from propagating faster than the local light speed \cite{HawkingEllis73}.
Bernig, Faifman \& Solanes' theory of pseudo-Riemannian curvature measures has potential relevance to such stability questions \cite{BernigFaifmanSolanes21} \cite{BernigFaifmanSolanes22a}.
 While the following example can be excluded 
either by requiring the metric-measure spacetimes to be globally hyperbolic regular Lorentzian length spaces (thus excluding future or past boundaries),
or by insisting $M_\infty = \spt m_\infty$,  it strongly suggests such modifications cannot restore compactness of $\NC^e_q(\infty)$ in the pointed measured sense,
even if there remains hope either for the analogous compactness of $\TCD^e_q(K,N)$ when $N<\infty$ \cite{Mueller22+} or for precompactness of spaces satisfying \eqref{WEC}.

\begin{example}[Instability of dimensionless null energy condition]
Take $M = \{ (x^1,\ldots, x^{n}) \in \R^n \mid x^1 \in [-1,1]\}$ to be a closed slab equipped the Minkowski metric $g_{ij} dx^i dx^j = dx_1^2 - \sum_{i=2}^n dx_i^2$,  Euclidean metric
$\tilde g_{ij} dx^i dx^j = \sum_{i=1}^n dx_i^2$, and
associated time-separation function $\ell=\ell_g$, distance $d=d_{\tilde g}$, and volume $\vg$.  
Fix the point $(x^1,\ldots,x^n)_j =0$ for all $j$. Only the reference measure $dm_j(x) = \exp (c_j-V_j(x)) d\vg(x)$ will vary along the sequence.
Taking $V_j(x) = - j g(x,x)$ yields Bakry-Emergy tensor $\Rc^{(\infty,V_j)}_{ab} = - j g_{ab}$ whose timelike bound from below is $-j$,
so $(M,d,\ell,m,\bar x)_j \in \NC_q^e(\infty)$ if $j<\infty$ by Remark~\ref{R:NECboundsT}.
On the other hand, choosing $c_j$ so that $m_j$ is a probability measure
ensures $m_j$ converges narrowly to $m_\infty = \frac12 [\delta_z + \delta_{-z}]$,  because $V_j(\pm x)$ attains its maxima only at $z=(1,0,\ldots,0)$.
But $(M,d,\ell,m,\bar x) \not \in \NC_q^e(\infty)$ since $x(t) = (2t-1)z$ makes the entropy of $\mu_t = \delta_{x(t)}$ finite precisely at its endpoints, 
and $(\mu_t)_{t \in [0,1]}$ is the unique $\ell_q$-path joining these endpoints because $x(t)$ is the unique $\ell$-path joining $x(0)=-z$ to $x(1)=z$.
\end{example}

A very interesting open question would be to determine whether a nonsmooth analog of the Penrose singularity theorem holds for $\NC^e_q(N)$
spaces.
Just as Cavalletti \& Mondino's analog of Hawking's theorem proves that big bang type singularities are an unavoidable consequence of instantaneous expansion even in nonsmooth models of the universe,  so an analog of the Penrose theorem would establish that even in the nonsmooth setting of $\NC^e_q(N)$ causal geodesic spaces,  null trapped surfaces must inevitably lead to the incomplete null geodesics signalling singularities stemming from stellar collapse.  Such a theorem has already been established in the manifold setting with a $C^1$ metric tensor $g$ by Graf~\cite{Graf20} (and for $g \in C^{1,1}$ in the earlier works that she cites).  Simultaneously and independently of the present work, Ketterer \cite{Ketterer23+} has shown in the smooth setting that both Hawking's area monotonicity \cite{Hawking72} and the Penrose singularity theorem follow from the displacement convexity of R\'enyi's power law entropy for measures on null geodesics.  
He has also shown in the smooth setting that this convexity gives another characterization of the null convergence condition.   Whether similar ideas extend to the nonsmooth setting is an intriguing open question.

\section*{Declarations}

{\bf Data Availability:} This article does not involve the use of any data.

\medskip\noindent
{\bf Competing Interests:} The author has no competing interests to declare that are relevant to the content of this article.

\medskip\noindent
{\bf Funding:} The author's research is supported in part by the Canada Research Chairs program CRC-2020-00289, the Simons Foundation, Natural Sciences and Engineering Research Council of Canada Discovery Grant RGPIN- 2020--04162, and Toronto's Fields Institute for the Mathematical Sciences, where part of this work was performed.

\bibliographystyle{plain}
\bibliography{../newbib.bib}

\begin{thebibliography}{10}

\bibitem{AkdemirCavallettiColinetMcCannSantarcangelo21}
Afiny Akdemir, Andrew Colinet, Robert McCann, Fabio Cavalletti, and Flavia
  Santarcangelo.
\newblock Independence of synthetic curvature dimension conditions on transport
  distance exponent.
\newblock {\em Trans. Amer. Math. Soc.}, 374(8):5877--5923, 2021.

\bibitem{AlexanderBishop08}
Stephanie~B. Alexander and Richard~L. Bishop.
\newblock Lorentz and semi-{R}iemannian spaces with {A}lexandrov curvature
  bounds.
\newblock {\em Comm. Anal. Geom.}, 16(2):251--282, 2008.

\bibitem{AllenBurtscher22}
Brian Allen and Annegret Burtscher.
\newblock Properties of the null distance and spacetime convergence.
\newblock {\em Int. Math. Res. Not. IMRN}, (10):7729--7808, 2022.

\bibitem{AmbrosioGigli13o}
Luigi Ambrosio and Nicola Gigli.
\newblock A user's guide to optimal transport.
\newblock In {\em Modelling and optimisation of flows on networks}, volume 2062
  of {\em Lecture Notes in Math.}, pages 1--155. Springer, Heidelberg, 2013.

\bibitem{AmbrosioGigliSavare14d}
Luigi Ambrosio, Nicola Gigli, and Giuseppe Savar\'{e}.
\newblock Metric measure spaces with {R}iemannian {R}icci curvature bounded
  from below.
\newblock {\em Duke Math. J.}, 163(7):1405--1490, 2014.

\bibitem{AnderssonHoward98}
Lars Andersson and Ralph Howard.
\newblock Comparison and rigidity theorems in semi-{R}iemannian geometry.
\newblock {\em Comm. Anal. Geom.}, 6(4):819--877, 1998.

\bibitem{BacherSturm10}
Kathrin Bacher and Karl-Theodor Sturm.
\newblock Localization and tensorization properties of the curvature-dimension
  condition for metric measure spaces.
\newblock {\em J. Funct. Anal.}, 259:28--56, 2010.

\bibitem{BakryEmery85}
D.~Bakry and Michel \'Emery.
\newblock Diffusions hypercontractives.
\newblock In {\em S\'eminaire de probabilit\'es, {XIX}, 1983/84}, volume 1123
  of {\em Lecture Notes in Math.}, pages 177--206. Springer, Berlin, 1985.

\bibitem{BeemEhrlichEasley96}
John~K. Beem, Paul~E. Ehrlich, and Kevin~L. Easley.
\newblock {\em Global {L}orentzian geometry}, volume 202 of {\em Monographs and
  Textbooks in Pure and Applied Mathematics}.
\newblock Marcel Dekker, Inc., New York, second edition, 1996.

\bibitem{BeranOctet23+}
Tobias Beran, Mathias Braun, Matteo Calisto, Nicola Gigli, Robert~J. McCann,
  Argam Ohanyan, Felix Rott, and Clemens S{\"a}mann.
\newblock In preparation.

\bibitem{BernigFaifmanSolanes21}
Andreas Bernig, Dmitry Faifman, and Gil Solanes.
\newblock Uniqueness of curvature measures in pseudo-{R}iemannian geometry.
\newblock {\em J. Geom. Anal.}, 31(12):11819--11848, 2021.

\bibitem{BernigFaifmanSolanes22a}
Andreas Bernig, Dmitry Faifman, and Gil Solanes.
\newblock Curvature measures of pseudo-{R}iemannian manifolds.
\newblock {\em J. Reine Angew. Math.}, 788:77--127, 2022.

\bibitem{BombelliNoldus04}
Luca Bombelli and Johan Noldus.
\newblock The moduli space of isometry classes of globally hyperbolic
  spacetimes.
\newblock {\em Classical Quantum Gravity}, 21(18):4429--4453, 2004.

\bibitem{Braun22.5+}
Mathias Braun.
\newblock Good geodesics satisfying the timelike curvature dimension condition.
\newblock {\em Nonlinear Anal.}, 229(113205), 2023.

\bibitem{Braun22.6+}
Mathias Braun.
\newblock {R\'enyi's entropy on Lorentzian spaces. Timelike curvature dimension
  conditions}.
\newblock {\em J. Math. Pures Appl.}, 177(9):46--128, 2023.

\bibitem{BraunHabermannSturm21}
Mathias Braun, Karen Habermann, and Karl-Theodor Sturm.
\newblock Optimal transport, gradient estimates, and pathwise {B}rownian
  coupling on spaces with variable {R}icci bounds.
\newblock {\em J. Math. Pures Appl. (9)}, 147:60--97, 2021.

\bibitem{BraunMcCann23p}
Mathias Braun and Robert~J. McCann.
\newblock {\em {\rm In preparation}}.

\bibitem{BuragoBuragoIvanov01}
Dmitri Burago, Yuri Burago, and Sergei Ivanov.
\newblock {\em A course in metric geometry}, volume~33 of {\em Graduate Studies
  in Mathematics}.
\newblock American Mathematical Society, Providence, RI, 2001.

\bibitem{BurtscherGarcia-Heveling21+}
Annegret Burtscher and Leonardo Garc\'{\i}a-Heveling.
\newblock Time functions on {L}orentzian length spaces.
\newblock {\em arXiv:2108.02693}.

\bibitem{BurtscherKettererMcCannWoolgar20}
Annegret Burtscher, Christian Ketterer, Robert~J. McCann, and Eric Woolgar.
\newblock Inscribed radius bounds for lower {R}icci bounded metric measure
  spaces with mean convex boundary.
\newblock {\em SIGMA Symmetry Integrability Geom. Methods Appl.}, 16:Paper No.
  131, 29, 2020.

\bibitem{Carroll04}
Sean Carroll.
\newblock {\em Spacetime and geometry}.
\newblock Addison Wesley, San Francisco, CA, 2004.
\newblock An introduction to general relativity.

\bibitem{Case10}
Jeffrey~S. Case.
\newblock Singularity theorems and the {L}orentzian splitting theorem for the
  {B}akry-{E}mery-{R}icci tensor.
\newblock {\em J. Geom. Phys.}, 60:477--490, 2010.

\bibitem{CavallettiMilman21}
Fabio Cavalletti and Emanuel Milman.
\newblock The globalization theorem for the curvature-dimension condition.
\newblock {\em Invent. Math.}, 226(1):1--137, 2021.

\bibitem{CavallettiMondino20+}
Fabio Cavalletti and Andrea Mondino.
\newblock {Optimal transport in Lorentzian synthetic spaces, synthetic timelike
  Ricci curvature lower bounds and applications}.
\newblock {\em {\rm Preprint at} arxiv.org/abs/2004.08934}.

\bibitem{CavallettiMondino22}
Fabio Cavalletti and Andrea Mondino.
\newblock A review of {L}orentzian synthetic theory of timelike {R}icci
  curvature bounds.
\newblock {\em Gen. Relativity Gravitation}, 54(11):Paper No. 137, 39, 2022.

\bibitem{DengPhD21}
Qin Deng.
\newblock {\em Holder {C}ontinuity of {T}angent {C}ones and {N}on-{B}ranching
  in {RCD}({K},{N}) {S}paces}.
\newblock ProQuest LLC, Ann Arbor, MI, 2021.
\newblock Thesis (Ph.D.)--University of Toronto (Canada).

\bibitem{EcksteinMiller17}
Micha\l\ Eckstein and Tomasz Miller.
\newblock Causality for nonlocal phenomena.
\newblock {\em Ann. Henri Poincar\'e}, 18:3049--3096, 2017.

\bibitem{EichmairGallowayPollack13}
Michael Eichmair, Gregory~J. Galloway, and Daniel Pollack.
\newblock Topological censorship from the initial data point of view.
\newblock {\em J. Differential Geom.}, 95(3):389--405, 2013.

\bibitem{ErbarKuwadaSturm15}
Matthias Erbar, Kazumasa Kuwada, and Karl-Theodor Sturm.
\newblock On the equivalence of the entropic curvature-dimension condition and
  {B}ochner's inequality on metric measure spaces.
\newblock {\em Invent. Math.}, 201:993--1071, 2015.

\bibitem{FeldmanMcCann02u}
M.~Feldman and R.J. McCann.
\newblock {Uniqueness and transport density in Monge's transportation problem}.
\newblock {\em Calc. Var. Partial Differential Equations}, 15:81--113, 2002.

\bibitem{Galloway00}
Gregory~J. Galloway.
\newblock Maximum principles for null hypersurfaces and null splitting
  theorems.
\newblock {\em Ann. Henri Poincar\'{e}}, 1(3):543--567, 2000.

\bibitem{GallowayKhuriWoolgar21}
Gregory~J. Galloway, Marcus~A. Khuri, and Eric Woolgar.
\newblock A {B}akry-\'{E}mery almost splitting result with applications to the
  topology of black holes.
\newblock {\em Comm. Math. Phys.}, 384(3):2067--2101, 2021.

\bibitem{Geroch68}
Robert Geroch.
\newblock Spinor structure of space-times in general relativity. {I}.
\newblock {\em J. Mathematical Phys.}, 9:1739--1744, 1968.

\bibitem{GigliMondinoSavare15}
Nicola Gigli, Andrea Mondino, and Giuseppe Savar\'{e}.
\newblock Convergence of pointed non-compact metric measure spaces and
  stability of {R}icci curvature bounds and heat flows.
\newblock {\em Proc. Lond. Math. Soc. (3)}, 111(5):1071--1129, 2015.

\bibitem{GigliRajalaSturm16}
Nicola Gigli, Tapio Rajala, and Karl-Theodor Sturm.
\newblock Optimal maps and exponentiation on finite-dimensional spaces with
  {R}icci curvature bounded from below.
\newblock {\em J. Geom. Anal.}, 26(4):2914--2929, 2016.

\bibitem{Graf20}
Melanie Graf.
\newblock Singularity theorems for {$C^1$}-{L}orentzian metrics.
\newblock {\em Comm. Math. Phys.}, 378(2):1417--1450, 2020.

\bibitem{GuilleminPollack74}
Victor Guillemin and Alan Pollack.
\newblock {\em Differential topology}.
\newblock Prentice-Hall, Inc., Englewood Cliffs, NJ, 1974.

\bibitem{Hawking72}
S.~W. Hawking.
\newblock Black holes in general relativity.
\newblock {\em Comm. Math. Phys.}, 25:152--166, 1972.

\bibitem{HawkingEllis73}
S.~W. Hawking and G.~F.~R. Ellis.
\newblock {\em The large scale structure of space-time}.
\newblock Cambridge University Press, London-New York, 1973.
\newblock Cambridge Monographs on Mathematical Physics, No. 1.

\bibitem{Ketterer15+}
Christian Ketterer.
\newblock Evolution variational inequality and {W}asserstein control in
  variable curvature context.
\newblock {\em arXiv:1509.02178}, pages 1--27, 2015.

\bibitem{Ketterer17}
Christian Ketterer.
\newblock On the geometry of metric measure spaces with variable curvature
  bounds.
\newblock {\em J. Geom. Anal.}, 27(3):1951--1994, 2017.

\bibitem{Ketterer23+}
Christian Ketterer.
\newblock Characterization of null energy via displacement convexity of
  entropy.
\newblock {\em arXiv:2304.01853}, pages 1--21, 2023.

\bibitem{KronheimerPenrose67}
E.~H. Kronheimer and R.~Penrose.
\newblock On the structure of causal spaces.
\newblock {\em Proc. Cambridge Philos. Soc.}, 63:481--501, 1967.

\bibitem{KunzingerSaemann18}
Michael Kunzinger and Clemens S\"{a}mann.
\newblock Lorentzian length spaces.
\newblock {\em Ann. Global Anal. Geom.}, 54(3):399--447, 2018.

\bibitem{KunzingerSteinbauer22}
Michael Kunzinger and Roland Steinbauer.
\newblock Null distance and convergence of {L}orentzian length spaces.
\newblock {\em Ann. Henri Poincar\'{e}}, 23(12):4319--4342, 2022.

\bibitem{Landsman21book}
Klaas Landsman.
\newblock {\em Foundations of general relativity, from Einstein to black
  holes}.
\newblock Radboud University Press, The Netherlands, 2021.

\bibitem{Landsman21}
Klaas Landsman.
\newblock Singularities, black holes, and cosmic censorship: a tribute to
  {R}oger {P}enrose.
\newblock {\em Found. Phys.}, 51(2):Paper No. 42, 38, 2021.
\newblock With an appendix by Erik Curiel.

\bibitem{LottVillani09}
John Lott and C\'{e}dric Villani.
\newblock Ricci curvature for metric-measure spaces via optimal transport.
\newblock {\em Ann. of Math. (2)}, 169(3):903--991, 2009.

\bibitem{McCann20}
Robert~J. McCann.
\newblock Displacement convexity of {B}oltzmann's entropy characterizes the
  strong energy condition from general relativity.
\newblock {\em Camb. J. Math.}, 8(3):609--681, 2020.

\bibitem{Minguzzi23+}
E.~Minguzzi.
\newblock {Further observations on the definition of global hyperbolicity under
  low regularity}.
\newblock {\em arXiv:2302.09284}.

\bibitem{MinguzziSuhr22+}
E.~Minguzzi and S.~Suhr.
\newblock {Lorentzian metric spaces and their Gromov-Hausdorff convergence}.
\newblock {\em arXiv:2209.14384}.

\bibitem{MondinoSuhr23}
Andrea Mondino and Stefan Suhr.
\newblock An optimal transport formulation of the {E}instein equations of
  general relativity.
\newblock {\em J. Eur. Math. Soc. (JEMS)}, 25(3):933--994, 2023.

\bibitem{Mueller22+}
Olaf Mueller.
\newblock {Gromov-Hausdorff distances for Lorentzian length spaces}.
\newblock {\em arXiv:2209.12736}.

\bibitem{Noldus02}
Johan Noldus.
\newblock A new topology on the space of {L}orentzian metrics on a fixed
  manifold.
\newblock {\em Classical Quantum Gravity}, 19(23):6075--6107, 2002.

\bibitem{Noldus04b}
Johan Noldus.
\newblock The limit space of a {C}auchy sequence of globally hyperbolic
  spacetimes.
\newblock {\em Classical Quantum Gravity}, 21(4):851--874, 2004.

\bibitem{Noldus04a}
Johan Noldus.
\newblock A {L}orentzian {G}romov-{H}ausdorff notion of distance.
\newblock {\em Classical Quantum Gravity}, 21(4):839--850, 2004.

\bibitem{NomizuOzeki61}
Katsumi Nomizu and Hideki Ozeki.
\newblock The existence of complete {R}iemannian metrics.
\newblock {\em Proc. Amer. Math. Soc.}, 12:889--891, 1961.

\bibitem{Penrose65a}
Roger Penrose.
\newblock Gravitational collapse and space-time singularities.
\newblock {\em Phys. Rev. Lett.}, 14:57--59, 1965.

\bibitem{RajalaSturm14}
Tapio Rajala and Karl-Theodor Sturm.
\newblock Non-branching geodesics and optimal maps in strong
  {$CD(K,\infty)$}-spaces.
\newblock {\em Calc. Var. Partial Differential Equations}, 50(3-4):831--846,
  2014.

\bibitem{SormaniVega16}
Christina Sormani and Carlos Vega.
\newblock Null distance on a spacetime.
\newblock {\em Classical Quantum Gravity}, 33(8):085001, 29, 2016.

\bibitem{Sturm06a}
Karl-Theodor Sturm.
\newblock On the geometry of metric measure spaces. {I}.
\newblock {\em Acta Math.}, 196(1):65--131, 2006.

\bibitem{Sturm06b}
Karl-Theodor Sturm.
\newblock On the geometry of metric measure spaces. {II}.
\newblock {\em Acta Math.}, 196(1):133--177, 2006.

\bibitem{Sturm15}
Karl-Theodor Sturm.
\newblock Metric measure spaces with variable {R}icci bounds and couplings of
  {B}rownian motions.
\newblock In {\em Festschrift {M}asatoshi {F}ukushima}, volume~17 of {\em
  Interdiscip. Math. Sci.}, pages 553--575. World Sci. Publ., Hackensack, NJ,
  2015.

\bibitem{Sturm20}
Karl-Theodor Sturm.
\newblock Distribution-valued {R}icci bounds for metric measure spaces,
  singular time changes, and gradient estimates for {N}eumann heat flows.
\newblock {\em Geom. Funct. Anal.}, 30(6):1648--1711, 2020.

\bibitem{Villani03}
C.~Villani.
\newblock {\em Topics in Optimal Transportation}, volume~58 of {\em Graduate
  Studies in Mathematics}.
\newblock American Mathematical Society, Providence, 2003.

\bibitem{Villani09}
C.~Villani.
\newblock {\em Optimal Transport. Old and New}, volume 338 of {\em Grundlehren
  der Mathematischen Wissenschaften [Fundamental Principles of Mathematical
  Sciences]}.
\newblock Springer, New York, 2009.

\bibitem{Woolgar13}
Eric Woolgar.
\newblock {Scalar-tensor gravitation and the {B}akry-{\'E}mery-{R}icci tensor}.
\newblock {\em Classical Quantum Gravity}, 30:085007, 8, 2013.

\bibitem{WoolgarWylie16}
Eric Woolgar and William Wylie.
\newblock {Cosmological singularity theorems and splitting theorems for
  {$N$}-{B}akry-{\'E}mery spacetimes}.
\newblock {\em J. Math. Phys.}, 57:022504, 12, 2016.

\end{thebibliography}

\end{document}